\documentclass[11pt]{amsart}
\usepackage{amsthm}
\usepackage{amsmath,amstext,amsthm,amsfonts,amssymb,amsthm}
\usepackage{multicol}
\usepackage{latexsym}
\usepackage{color}
\usepackage{graphicx}
\usepackage{epsf}
\usepackage{epsfig}
\usepackage{epic}
\usepackage{graphics}
\usepackage{verbatim}
\usepackage{color}
\usepackage{stmaryrd}
\usepackage{hyperref}
\newtheorem{theorem}{Theorem}[section]
\newtheorem{lemma}[theorem]{Lemma}
\theoremstyle{definition}
\newtheorem{definition}[theorem]{Definition}
\newtheorem{proposition}[theorem]{Proposition}

\theoremstyle{remark}
\newtheorem{remark}[theorem]{Remark}
\numberwithin{equation}{section}

\newcommand{\quat}{\mathbb H}

\newcommand{\N}{\mathfrak{N}}
\newcommand{\an}{\mathfrak{a}}

\newcommand{\q}{\mathbf{q}}

\newcommand{\HI}{\mathfrak{H}}

\newcommand{\D}{\mathfrak{D}}

\newcommand{\oz}{\overline{z}}
\newcommand{\qu}{\mathfrak{q}}
\newcommand{\pu}{\mathfrak{p}}
\newcommand{\as}{\mathsf{a}}
\newcommand{\asd}{\mathsf{a}^{\dagger}}
\newcommand{\mx}{\mathfrak{q}}
\newcommand{\oqu}{\overline{\mathfrak{q}}}

\begin{document}
\title[Weyl-Heisenberg Lie Algebra ]{
	A Representation of Weyl-Heisenberg Lie Algebra in the Quaternionic Setting}
\author{B. Muraleetharan$^{\dagger}$, K. Thirulogasanthar$^{\ddagger}$, I. Sabadini$^\star$}
\address{$^{\dagger}$ Department of mathematics and Statistics, University of Jaffna, Thirunelveli, Sri Lanka.}
\address{$^{\ddagger}$ Department of Computer Science and Software Engineering, Concordia University, 1455 De Maisonneuve Blvd. West, Montreal, Quebec, H3G 1M8, Canada.}
\address{$^*$ Dipartimento di Matematica, Politecnico di Milano, Via E. Bonardi, 9, 20133, Milano, Italy.}
\email{santhar@gmail.com, bbmuraleetharan@jfn.ac.lk and irene.sabadini@polimi.it}
\date{\today}
\thanks{K. Thirulogasanthar would like to thank the FQRNT, Fonds de la Recherche  Nature et  Technologies (Quebec, Canada) for partial financial support under the grant number 2017-CO-201915. Part of this work was done while he was visiting the Politecnico di Milano to which he expresses his thanks for the hospitality. He also thanks the program Professori Visitatori GNSAGA, INDAM for the support during the period in which this paper was partially written.}
\pagestyle{myheadings}
\begin{abstract}
 Using a left multiplication defined on a right quaternionic Hilbert space, linear self-adjoint momentum operators on a right quaternionic Hilbert space are defined in complete analogy with their complex counterpart. With the aid of the so-obtained position and momentum operators, we study the Heisenberg uncertainty principle on the whole set of quaternions and on a quaternionic slice, namely on a copy of the complex plane inside the quaternions. For the quaternionic harmonic oscillator, the uncertainty relation is shown to saturate on a neighborhood of the origin in the case we consider the whole set of quaternions, while it is saturated on the whole slice in the case we take the slice-wise approach. In analogy with  the complex Weyl-Heisenberg Lie algebra, Lie algebraic structures are developed for the quaternionic case. Finally, we introduce a quaternionic displacement operator which is square integrable, irreducible and unitary,  and we study its properties.
\end{abstract}
\maketitle
\section{Introduction}\label{sec_intro}
In the past few years there has been a resurgence of interest for the quaternionic quantum mechanics. This topic was extensively studied starting from the celebrated paper of Birkhoff and von Neumann who asserted that quantum mechanics can be studied only over the complex and the quaternionic numbers, see \cite{bvn}. It culminated with the book of Adler in 1995 and, after that, the interest on this topic started fading, indeed some crucial ingredients of the theory were missing and prevented further developments of the subject.

As it is well known, quaternions can always be represented as a pair of complex numbers (the so-called symplectic components) and hence quaternions possess a symplectic structure. As in the complex quantum mechanics, states of quaternionic quantum mechanics are represented by vectors of a separable quaternionic Hilbert space and observables are represented by quaternionic linear and self-adjoint operators. However, quaternionic quantum mechanics is different from the complex quantum mechanics \cite{Ad}.\\

There are three main issues which prevented physicists and mathematicians to construct a well-formed quaternionic quantum mechanics. In this introduction we will discuss them and we will show that, while the answer to one of them was found in the past few years, the other two issues are addressed in this manuscript.\\

A first missing ingredient was an appropriate spectral theory for quaternionic linear operators. Several papers treated the spectral theorem, see e.g. \cite{fjss}, but, surprisingly, the notion of spectrum in use was not specified. Physicists (see \cite{Ad} and references therein) were considering the right eigenvalue problem $A\phi=\phi\qu$ which, however, is not associated with a quaternionic linear operator when $\qu$ is a quaternion, and thus does not allow an appropriate definition of spectrum.\\
The appropriate notion of spectrum, the so-called $S$-spectrum, was introduced more than twenty years after the book of Adler, see e.g. \cite{NFC} for the basic results on the $S$-spectrum and the functional calculus associated with it. Given an operator $A$, its $S$-spectrum is defined as the set of quaternions $q$ for which the second order operator $A^2 -2 {\rm Re}(\qu)A +|\qu|^2 I$ is not invertible.
This second order operator is the linear operator which, when one restricts to the point spectrum, gives the set of right eigenvalues and thus generalizes the notion of eigenvalue used by physicists.
Using this spectrum one can prove a spectral theorem for quaternionic normal operators, bounded or unbounded, see \cite{ack}.\\

The second important issue is related to the notion of a universal momentum operator. There is no quaternionic linear self-adjoint momentum operator analogous to that one of the complex case. In fact, if one tries to mimic the complex momentum coordinate, using one of the quaternionic imaginary units $i, j$ or $k$, then the corresponding momentum operator becomes non self-adjoint. This is due to the fact that, in general, for any quaternionic linear operator $A$ and any quaternion $\qu\in\quat$, $(\qu A)^{\dagger}\not=\overline{\qu}~A^{\dagger}$. For a detailed explanation we refer the reader to \cite{MuTh}. Further, according to \cite{Ad}, there is no quaternionic self-adjoint momentum operator having all the properties fulfilled by the momentum operator of complex quantum mechanics. For various attempts to solve this problem and their drawbacks see \cite{Ad} (pages 52-64).\\
In this paper we propose a solution of this issue.
To explain the strategy, let us recall that since quaternions do not commute, there are three types of quaternionic Hilbert spaces: left, right and two-sided, according to how vectors are multiplied by scalars. Most of the linear spaces are one-sided, but in order to have suitable properties on the linear operators acting on them, it is necessary to have a notion of multiplication on both sides. As it is well known, given a right quaternionic Hilbert space it is possible to equip it with a left multiplication. This is the main tool we use to solve the second issue. In fact, in such a space, we shall define linear self-adjoint momentum operator in complete analogy with the complex momentum operator. \\

It ought to be pointed out that the left multiplication on a right quaternionic Hilbert space is basis dependent, that is, the multiplication is defined in terms of a basis of the Hilbert space. However, the basis dependence is not a problem when we study a particular physical system because in this case we always work in the state Hilbert space, which is usually taken as the Hilbert space spanned by the wavefunctions of the Hamiltonian (the Fock space).\\

Using the momentum operator obtained with our strategy, we shall study the Heisenberg uncertainty principle. We will show that, if we work on the whole set of quaternions $\quat$, the uncertainty relation gets saturated in a limit sense in a neighborhood of the origin. This fact happens due to the non-commutativity of quaternions  which prevented us in getting a closed form for a series. However, if we consider a quaternionic slice-wise approach, namely if we work on copies of the complex plane contained in $\mathbb H$, the uncertainty relation holds exactly as for the complex harmonic oscillator.\\

We shall also consider another thorny issue of quaternionic quantum mechanics, which is the harmonic oscillator displacement operator. In the recent years, this operator was discussed twice in the literature on quaternionic quantum mechanics: first in \cite{Ad2} and then in \cite{Thi2}. It has been proved that there is an operator analogous to the complex displacement operator that can generate quaternionic canonical coherent states by the action on the ground state in a quaternionic Hilbert space. However, it fails to be a group representation of Weyl-Heisenberg algebra (or any other algebra) \cite{Thi2}. In this paper, we shall show that with a left multiplication on a right quaternionic Hilbert space such a displacement operator leads to a square integrable, irreducible and unitary representation. Further, it satisfies almost all the properties of its complex counterpart.\\

Finally, we mention another fact which is not considered in this paper but belongs to the same mathematical framework. In complex quantum mechanics, the tensor product composition is well-defined, but there has been historical difficulty in sensibly defining system composition using quaternions, see e.g. \cite{NJ}. In \cite{acs}, section 4 and \cite{Al}, it is shown how tensor products can be constructed over the quaternions. This is, in our opinion, another point in favor of the proposed mathematical realm where to consider quaternionic quantum mechanics.\\

The article is organized as follows. In section 2 we revisit the quaternions and gather some material pertinent to the development of the manuscript. In particular, we review the left multiplication on a right quaternionic Hilbert space. In section 3 we recall the quantization of quaternions based on \cite{MuTh} but with the left multiplication defined in section 2. In section 3.1 we discuss quaternionic coherent states on a right quaternionic Hilbert space. In section 3.2 the coherent state quantization of quaternions is presented. Section 3.3 deals with number, position, momentum operator and the quaternionic harmonic oscillator. In particular, we define a class of quaternionic linear self-adjoint momentum operators. In section 3.4 we study the quaternionic Heisenberg uncertainty on the whole set of quaternions and on a quaternionic slice. Section 4 is devoted to the Lie algebraic structures of the operators defined in section 3, in particular, section 4.1 defines some Lie algebraic structures. In section 4.2 we obtain the Weyl-Heisenberg Lie algebra in a quaternionic setting. In section 4.3 we demonstrate the existence of a square integrable, irreducible and unitary harmonic oscillator displacement operator. In section 5.1 we discuss some symmetry properties of the operators and 5.2 deals with some properties of the displacement operator defined in section 4.3. Section 6 ends the manuscript with a conclusion.\\
\section{Mathematical preliminaries}
In order to make the paper self-contained, we recall a few facts about quaternions which may not be well-known. In particular, we revisit the $2\times 2$ complex matrix representations of quaternions, quaternionic Hilbert spaces as needed here. For details we refer the reader to \cite{Ad, Albook, Gen1, Vis, Za}.
\subsection{Quaternions}
Let $\quat$ denote the field of quaternions. Its elements are of the form $\qu=q_0+q_1i+q_2j+q_3k$ where $q_0,q_1,q_2$ and $q_3$ are real numbers, and $i,j,k$ are imaginary units such that $i^2=j^2=k^2=-1$, $ij=-ji=k$, $jk=-kj=i$ and $ki=-ik=j$. The quaternionic conjugate of $\qu$ is defined to be $\overline{\qu} = q_0 - q_1i - q_2j - q_3k$. We shall find it convenient to use the representation of quaternions by $2\times 2$ complex matrices:
\begin{equation}
\qu = q_0 \sigma_{0} + i \q \cdot \underline{\sigma},
\end{equation}
with $q_0 \in \mathbb R , \quad \q = (q_1, q_2, q_3)
\in \mathbb R^3$, $\sigma_0 = \mathbb{I}_2$, the $2\times 2$ identity matrix, and
$\underline{\sigma} = (\sigma_1, -\sigma_2, \sigma_3)$, where the
$\sigma_\ell, \; \ell =1,2,3$ are the usual Pauli matrices. The quaternionic imaginary units are identified as, $i = \sqrt{-1}\sigma_1, \;\; j = -\sqrt{-1}\sigma_2, \;\; k = \sqrt{-1}\sigma_3$. Thus,
\begin{equation}
\qu = \left(\begin{array}{cc}
q_0 + i q_3 & -q_2 + i q_1 \\
q_2 + i q_1 & q_0 - i q_3
\end{array}\right) \qquad 
\label{q3}
\end{equation}
and $\overline{\qu} = \qu^\dag\quad \text{(matrix adjoint)}\; .$
Introducing the polar coordinates:
\begin{eqnarray*}
	q_0 &=& r \cos{\theta}, \\
	q_1 &=& r \sin{\theta} \sin{\phi} \cos{\psi}, \\
	q_2 &=& r \sin{\theta} \sin{\phi} \sin{\psi}, \\
	q_3 &=& r \sin{\theta} \cos{\phi},
\end{eqnarray*}
where $(r,\phi,\theta,\psi)  \in [0,\infty)\times[0,\pi]\times[0,2\pi)^{2}$, we may write
\begin{equation}
\qu = A(r) e^{i \theta \sigma(\widehat{n})}
\label{q4},
\end{equation}
where
\begin{equation}
A(r) = r\mathbb \sigma_0
\end{equation}
and\begin{equation}
\sigma(\widehat{n}) = \left(\begin{array}{cc}
\cos{\phi} & \sin{\phi} e^{i\psi} \\
\sin{\phi} e^{-i\psi} & -\cos{\phi}
\end{array}\right).
\label{q5}\end{equation}
The matrices
$A(r)$ and $\sigma(\widehat{n})$ satisfy the conditions,
\begin{equation}
A(r) = A(r)^\dagger,~\sigma(\widehat{n})^2 = \sigma_0
,~\sigma(\widehat{n})^\dagger = \sigma(\widehat{n})
\label{san1}
\end{equation}
and
$\lbrack A(r), \sigma(\widehat{n}) \rbrack = 0.$
Note that a real norm on $\quat$ is defined by  $$\vert\qu\vert^2  := \overline{\qu} \qu = r^2 \sigma_0 = (q_0^2 +  q_1^2 +  q_2^2 +  q_3^2)\mathbb{I}_2.$$ A typical measure on $\quat$ may take the form
\begin{equation}\label{measure}
d\varsigma(r, \theta,\phi, \psi)= d\tau(r)\, d\theta\, d\Omega(\phi ,\psi )
\end{equation}
with $d\Omega(\phi ,\psi) = \displaystyle{\frac{1}{4\pi}} \,\sin{\phi}\, d\phi \,d\psi .$
Note also that for ${\pu},\qu\in \quat$, we have $\overline{{\pu}\qu}=\overline{\qu}~\overline{{\pu}}$, $\pu\qu\not=\qu\pu$, $\qu\overline{{ \qu}}=\overline{{\qu}}\qu$, and real numbers commute with quaternions.
Quaternions can also be interpreted as a sum of scalar and a vector by writing $$\qu=q_0+q_1i+q_2j+q_3k=(q_0,\q);$$ where $\q=q_1i+q_2j+q_3k$. This particular method of writing quaternions will help us for forming a representation of Weyl-Heisenberg Lie algebra under the quaternionic settings. We borrow the materials as needed here from \cite{Gen1}.  Let
\begin{eqnarray*}
	\mathbb{S}&=&\{I=x_1i+x_2j+x_3k~\vert
	~x_1,x_2,x_3\in\mathbb{R},~x_1^2+x_2^2+x_3^2=1\},
\end{eqnarray*}
we call it a quaternionic sphere.
 \begin{proposition}\cite{Gen1}\label{Pr1}
	For any non-real quaternion $\qu\in \quat\smallsetminus\mathbb{R}$, there exist, and are unique, $x,y\in\mathbb{R}$ with $y>0$, and $I_\qu\in\mathbb{S}$ such that $\qu=x+I_\qu y$.
\end{proposition}
For every quaternion $I\in\mathbb{S}$, the complex plane $\mathbb{C}_I=\mathbb{R}+I\mathbb{R}$ passing through the origin, and containing $1$ and $I$, is called a quaternionic slice. Thereby, we can see that
\begin{equation}\label{Eq1}
\quat=\bigcup_{I\in\mathbb{S}}\mathbb{C}_I\quad\text{and}\quad\bigcap_{I\in\mathbb{S}} \mathbb{C}_I=\mathbb{R}
\end{equation}
One can also easily see that $\mathbb{C}_I\subset \quat$ is commutative, while, elements from two different quaternion slices, $\mathbb{C}_I$ and $\mathbb{C}_J$ (for $I, J\in\mathbb{S}$ with $I\not=J$), do not commute unless they are real.
\subsection{Quaternionic Hilbert spaces}
In this subsection we  define left and right quaternionic Hilbert spaces. For details we refer the reader to \cite{Ad, Albook}. We also define the Hilbert space of square integrable functions on quaternions based on \cite{Gu, Vis}.
\subsubsection{Right Quaternionic Hilbert Space}
Let $V_{\quat}^{R}$ be a linear vector space under right multiplication by quaternionic scalars (again $\quat$ standing for the field of quaternions).  For $f,g,h\in V_{\quat}^{R}$ and $\qu\in \quat$, the inner product
$$\langle\cdot\mid\cdot\rangle:V_{\quat}^{R}\times V_{\quat}^{R}\longrightarrow \quat$$
satisfies the following properties
\begin{enumerate}
	\item[(i)]
	$\overline{\langle f\mid g\rangle}=\langle g\mid f\rangle$
	\item[(ii)]
	$\|f\|^{2}=\langle f\mid f\rangle>0$ unless $f=0$, a real norm
	\item[(iii)]
	$\langle f\mid g+h\rangle=\langle f\mid g\rangle+\langle f\mid h\rangle$
	\item[(iv)]
	$\langle f\mid g\qu\rangle=\langle f\mid g\rangle\qu$
	\item[(v)]
	$\langle f\qu\mid g\rangle=\overline{\qu}\langle f\mid g\rangle$
\end{enumerate}
where $\overline{\qu}$ stands for the quaternionic conjugate. We assume that the
space $V_{\quat}^{R}$ is complete under the norm given above. Then,  together with $\langle\cdot\mid\cdot\rangle$, this defines a right quaternionic Hilbert space, which we shall assume to be separable. Quaternionic Hilbert spaces share most of the standard properties of complex Hilbert spaces. In particular, the Cauchy-Schwartz inequality holds on quaternionic Hilbert spaces as well as the Riesz representation theorem for their duals.  Thus, the Dirac bra-ket notation
can be adapted to quaternionic Hilbert spaces:
$$\mid f\qu\rangle=\mid f\rangle\qu,\hspace{1cm}\langle f\qu\mid=\overline{\qu}\langle f\mid\;, $$
for a right quaternionic Hilbert space, with $\vert f\rangle$ denoting the vector $f$ and $\langle f\vert$ its dual vector. Similarly the left quaternionic Hilbert space $V_{\quat}^{L}$ can also be described, see for more detail \cite{Ad,MuTh,Thi1}.
The field of quaternions $\quat$ itself can be turned into a left quaternionic Hilbert space by defining the inner product $\langle \qu \mid \qu^\prime \rangle = \qu \qu^{\prime\dag} = \qu\overline{\qu^\prime}$ or into a right quaternionic Hilbert space with  $\langle \qu \mid \qu^\prime \rangle = \qu^\dag \qu^\prime = \overline{\qu}\qu^\prime$. Further note that, due to the non-commutativity of quaternions the sum
$\sum_{m=0}^{\infty}\pu^m\qu^m/m!$
cannot be written as $\text{exp}(\pu\qu).$ However, in any Hilbert space the norm convergence implies the convergence of the series and
$\sum_{m=0}^{\infty}\left\vert \pu^m\qu^m/m!\right\vert
\leq e^{|\pu||\qu|},$ therefore $\sum_{m=0}^{\infty}\pu^m\qu^m/m!$ converges.
\subsubsection{Quaternionic Hilbert Spaces of Square Integrable Functions}
Let $(X, \mu)$ be a measure space and $\quat$  the field of quaternions, then
$$\left\{f:X\rightarrow \quat \left|  \int_X|f(x)|^2d\mu(x)<\infty \right.\right\}$$
\noindent
is a right quaternionic Hilbert space which is denoted by $L^2_\quat(X,\mu)$, with the (right) scalar product
\begin{equation}
\langle f \mid g\rangle =\int_X\overline{ f(x)}{g(x)} d\mu(x),
\label{left-sc-prod}
\end{equation}
where $\overline{f(x)}$ is the quaternionic conjugate of $f(x)$, and (right)  scalar multiplication $f\an, \; \an\in \quat,$ with $(f\an)(\qu) = f(\qu)\an$ (see \cite{Gu,Vis} for details). Similarly, one could define a left quaternionic Hilbert space of square integrable functions.
\subsection{Left Scalar Multiplication on $V_{\quat}^{R}$.}
The following idea of introducing a left multiplication in a right quaternionic Hilbert space goes back to Teichm\"uller, see \cite{Teich} p. 96, and it is also used in \cite{Vis}. In this section we shall extract from \cite{ghimorper} the definition and some properties of left scalar multiples of vectors on $V_{\quat}^R$  as needed for the development of the manuscript. The left scalar multiple of vectors on a right quaternionic Hilbert space is an extremely non-canonical operation associated with a choice of preferred Hilbert basis. Now the Hilbert space $V_{\quat}^{R}$ has a Hilbert basis
\begin{equation}\label{b1}
\mathcal{O}=\{\varphi_{k}\,\mid\,k\in \N\},
\end{equation}
where $\N$ is a countable index set.
The left scalar multiplication `$\cdot$' on $V_{\quat}^{R}$ induced by $\mathcal{O}$ is defined as the map $\quat\times V_{\quat}^{R}\ni(\qu,\phi)\longmapsto \qu\cdot\phi\in V_{\quat}^{R}$ given by
\begin{equation}\label{LPro}
\qu\cdot\phi:=\sum_{k\in \N}\varphi_{k}\qu\langle \varphi_{k}\mid \phi\rangle,
\end{equation}
for all $(\qu,\phi)\in\quat\times V_{\quat}^{R}$. Since all left multiplications are made with respect to some basis, assume that the basis $\mathcal{O}$ given by (\ref{b1}) is fixed all over the paper. However we shall make choices wherever needed.
\begin{proposition}\cite{ghimorper}\label{lft_mul}
	The left product defined in (\ref{LPro}) satisfies the following properties. For every $\phi,\psi\in V_{\quat}^{R}$ and $\pu,\qu\in\quat$,
	\begin{itemize}
		\item[(a)] $\qu\cdot(\phi+\psi)=\qu\cdot\phi+\qu\cdot\psi$ and $\qu\cdot(\phi\pu)=(\qu\cdot\phi)\pu$.
		\item[(b)] $\|\qu\cdot\phi\|=|\qu|\|\phi\|$.
		\item[(c)] $\qu\cdot(\pu\cdot\phi)=(\qu\pu\cdot\phi)$.
		\item[(d)] $\langle\overline{\qu}\cdot\phi\mid\psi\rangle
		=\langle\phi\mid\qu\cdot\psi\rangle$.
		\item[(e)] $r\cdot\phi=\phi r$, for all $r\in \mathbb{R}$.
		\item[(f)] $\qu\cdot\varphi_{k}=\varphi_{k}\qu$, for all $\varphi_{k}\in\mathcal O$, $k\in \N$.
	\end{itemize}
\end{proposition}
\begin{remark}\label{Rem123}
	One can trivially see that $(\pu+\qu)\cdot\phi=\pu\cdot\phi+\qu\cdot\phi$, for all $\pu,\qu\in\quat$ and $\phi\in V_{\quat}^{R}$. Moreover, with the aid of (b) in above Proposition (\ref{lft_mul}), we can have, if $\{\phi_n\}$ in $V_\quat^R$ is such that $\phi_n\longrightarrow\phi$, then $\qu\cdot\phi_n\longrightarrow\qu\cdot\phi$. Also if $\sum_{n}\phi_n$ is a convergent sequence in $V_\quat^R$, then $\qu\cdot(\sum_{n}\phi_n)=\sum_{n}\qu\cdot\phi_n$.
\end{remark}
Furthermore, the quaternionic scalar multiplication of $\quat$-linear operators is also defined in \cite{ghimorper}. For any fixed $\qu\in\quat$ and a given right $\quat$-linear operator $A:\D(A)\longrightarrow V_{\quat}^{R}$, the left scalar multiplication `$\cdot$' of $A$ is defined as a map $\qu \cdot A:\D(A)\longrightarrow V_{\quat}^{R}$ by setting
\begin{equation}\label{lft_mul-op}
(\qu\cdot A)\phi:=\qu\cdot (A\phi)=\sum_{k\in \N}\varphi_{k}\qu\langle \varphi_{k}\mid A\phi\rangle,
\end{equation}
for all $\phi\in \D(A)$. It is straightforward that $\qu A$ is a right $\quat$-linear operator. If $\qu\cdot\phi\in \D(A)$, for all $\phi\in \D(A)$, one can define right scalar multiplication `$\cdot$' of the right $\quat$-linear operator $A:\D(A)\longrightarrow V_{\quat}^{R}$ as a map $ A\cdot\qu:\D(A)\longrightarrow V_{\quat}^{R}$ by the setting
\begin{equation}\label{rgt_mul-op}
(A\cdot\qu )\phi:=A(\qu\cdot \phi),
\end{equation}
for all $\phi\in \D(A)$. The operator $A\cdot\qu$  is also a right $\quat$-linear operator. One can easily obtain that, if $\qu\cdot\phi\in \D(A)$ for all $\phi\in \D(A)$, and $\D(A)$ is dense in $V_{\quat}^{R}$, then
\begin{equation}\label{sc_mul_aj-op}
(\qu\cdot A)^{\dagger}=A^{\dagger}\cdot\overline{\qu}~\mbox{~and~}~
(A\cdot\qu)^{\dagger}=\overline{\qu}\cdot A^{\dagger}.
\end{equation}

\section{Quantization of the quaternions}
In physics, quantization is a procedure that associates with an algebra $A_{cl}$ of classical observables an algebra $A_q$ of quantum observables. The algebra $A_{cl}$ is usually realized as a commutative Poisson algebra of derivable functions on a symplectic (or phase) space $X$. The algebra $A_q$ is, however, non commutative in general and the quantization procedure must provide a correspondence $A_{cl}\mapsto A_q~:~f\mapsto A_f$. Most physical quantum theories may be obtained as the result of a canonical quantization procedure. However, among the various quantization procedures available in the literature, the coherent state quantization appear quite natural because the only structure that a space $X$ must possess is a measure. For various quantization procedures and their advantages and drawbacks we refer the reader to \cite{Ali, AE,Gaz}.
\subsection{Coherent states on right quaternionic Hilbert spaces}\label{CSLQH}
The main content of this section is extracted from \cite{Thi2} as needed here. For an enhanced explanation we refer the reader to \cite{Thi2}. In \cite{Thi2} the authors have defined coherent states on $V_{\quat}^{R}$ and $V_{\quat}^{L}$, and also established the normalization and resolution of the identities for each of them. We briefly revisit the coherent states of $V_{\quat}^{R}$ and the normalization and resolution of the identity.
Let $\{\mid e_{m}\rangle\}_{m=0}^{\infty}$ be an orthonormal basis of $V_{\quat}^{R}$. For $\qu\in \quat$, the coherent states (CS) are defined as vectors in $V_{\quat}^{R}$ of the form
\begin{equation}\label{CS}
\mid\qu\rangle=\mathcal N (\mid\qu\mid)^{-\frac{1}{2}}\sum_{m=0}^{\infty}\mid e_m\rangle \frac{\qu^{m}}{\sqrt{\rho(m)}},
\end{equation}
where $\mathcal N (\mid\qu\mid)$ is the normalization factor and $\{\rho(m)\}_{m=0}^{\infty}$ is a positive sequence of real numbers, see also \cite{MuTh} for more details.
The resolution of the identity is
\begin{equation}\label{res}
\int_{\mathcal D}\mid\qu\rangle\langle\qu\mid d\varsigma(r, \theta,\phi, \psi)=\mathbb{I}_{V_{\quat}^{R}},
\end{equation}
where $\mathbb{I}_{V_{\quat}^{R}}$ is the identity operator on $V_{\quat}^{R}$.
In particular, if $\rho(m)=m!$, then the normalization factor is $\mathcal N (\mid\qu\mid)=e^{|\qu|^{2}}$ and a resolution of the identity holds with the measure given in (\ref{measure}). For this choice, the vectors (CS) defined by (\ref{CS})  are called {\em right quaternionic canonical coherent states} (RQCS). For the purpose of quantizing the quaternions we shall use these canonical set of CS.\\

At this point a note is in order. By Proposition \ref{lft_mul} (f), for a basis element $|e_n\rangle$ we have $\qu\cdot |e_n\rangle=|e_n\rangle\qu$. For this reason in defining the coherent states we have avoided putting the left multiplication symbol. Since coherent states are linear combination of the basis vectors, for the same reason, we shall also avoid putting it in the definition of the quantization map.
\subsection{Quantization}
For a general scheme of quantization we refer the reader, for example, to \cite{Gaz, AE}. For the quantization of quaternions we refer to \cite{MuTh}. The material presented here, related to quantization of quaternions, is extracted from \cite{MuTh}.
Since $(\quat,d\varsigma(r, \theta,\phi, \psi))$ is a measure space, the set
\begin{equation*}
\left\{f:\quat\rightarrow \quat\mid\int_\quat|f(\qu)|^2d\varsigma(r, \theta,\phi, \psi)<\infty\right\}
\end{equation*}
is the space of right quaternionic square integrable functions and is denoted by  $L^2_{\quat}(\quat,d\varsigma(r, \theta,\phi, \psi))$. Define the sequence of functions $\{\phi_{n}\}_{n=0}^{\infty}$ such that $$\phi_{n}: \quat\longrightarrow \quat$$ by
\begin{equation}\label{phi}
\phi_{n}(\qu)=\frac{\overline{\qu}^{n}}{\sqrt{n!}},~~\mbox{~~for all~~}~~\qu\in \quat.\end{equation}
Then $\phi_{n}, \overline{\phi}_n\in L^2_{\quat}(\quat,d\varsigma(r, \theta,\phi, \psi))$, for all $n=0,1,2\cdots,$~~  $\langle\phi_{m}\mid\phi_{n}\rangle=\delta_{mn}$ and $\langle\overline{\phi}_{m}\mid\overline{\phi}_{n}\rangle=\delta_{mn}$ (see \cite{Thi2,MuTh}).
In other words, $$\mathcal{O}_{\text{ar}}=\{\phi_n~\vert~n=0,1,2\cdots\}\quad\text{and}\quad \mathcal{O}_{\text{r}}=\{\overline{\phi}_n~\vert~n=0,1,2\cdots\}$$ are orthonormal sets in $L^2_{\quat}(\quat,d\varsigma(r, \theta,\phi, \psi))$. The right quaternionic span of $\mathcal{O}_{\text{r}}$ and $\mathcal{O}_{\text{ar}}$ are the spaces of right-regular and the anti-right-regular functions respectively \cite{Thi1} (the counterparts of complex Bargmann spaces of holomorphic and anti-holomorphic functions). Let $\mathfrak{H}$ be a separable right quaternionic Hilbert space with an orthonormal basis  $$\mathcal{E}=\{~\mid e_{n}\rangle~\mid~ n=0,1,2\cdots~\}$$ which is in $1-1$ correspondence with $\mathcal{O}_{\text{r}}$ and with $\mathcal{O}_{\text{ar}}$. If needed, we can also take
$$\HI=\overline{\text{right-span-over}~\quat}~\mathcal{O}_{\text{r}}\quad\text{or}\quad \HI=\overline{\text{right-span-over}~\quat}~\mathcal{O}_{\text{ar}},$$
where the bar stands for the closure.
 Then the coherent states (\ref{CS}) become
\begin{equation}\label{CS1}
\mid\gamma_{\qu}\rangle=e^{-\mid\qu\mid^{2}/2}\sum_{m=0}^{\infty}\mid e_m\rangle\overline{\phi_{m}}.
\end{equation}
Using the set of CS (\ref{CS1}) we shall establish the coherent state quantization on $\mathfrak{H}$ by associating a function $$\quat\ni \qu\longmapsto f(\qu,\overline{\qu}).$$
Now let us define the operator on $\mathfrak{H}$ by
\begin{equation}
f(\qu,\overline{\qu})\mapsto A_{f},
\end{equation}
where $A_{f}$ is given by the operator valued integral
\begin{equation}\label{Qmap}
A_{f}=\int_{\quat}\mid\gamma_{\qu}\rangle f(\qu,\overline{\qu})\langle\gamma_{\qu}\mid  d\varsigma(r, \theta,\phi, \psi)=\sum_{m=0}^{\infty}\sum_{l=0}^{\infty}\frac{\mid e_{m}\rangle J_{m,l}\langle e_{l}\mid}{\sqrt{m!~l!}},
\end{equation}
where the integral $J_{m,l} $ is given by $$\displaystyle \iiiint\limits_{[0,\infty)\times[0,\pi]\times[0,2\pi)^{2}}\frac{ {\qu^{m}}f(\qu,\overline{\qu}){\overline{\qu}^{l}}}{e^{r^2}}d\varsigma(r, \theta,\phi, \psi).$$
By direct calculation we have that
if  $f(\qu,\overline{\qu})=\qu$, then
\begin{equation}\label{Aq}
A_{\qu}
=\sum_{m=0}^{\infty}\sqrt{(m+1)}\mid e_{m}\rangle\langle e_{m+1}\mid
\end{equation}
and if $f(\qu,\overline{\qu})=\overline{\qu}$, then
\begin{equation}\label{Aqb}
A_{\overline{\qu}}
=\sum_{m=0}^{\infty}\sqrt{(m+1)}\mid e_{m+1}\rangle\langle e_{m}\mid.
\end{equation}
Moreover if $f(\qu,\overline{\qu})=1$, then
$A_{1}
=\mathbb{I}_{\mathfrak{H}}.$ Since the operators $A_{\qu}$ and $A_{\oqu}$ are independent of $\qu$ and for the notational convenience, from now on we denote $A_{\qu} :=\mathsf{a}$ and $A_{\overline{\qu}} :=\mathsf{a}^{\dagger}$. This notation is consistent since
$$\langle {\mathsf{a}}^\dagger \phi\mid \psi\rangle=\langle \phi\mid {\mathsf{a}} \psi\rangle;\quad\text{for all}\quad |\phi\rangle, |\psi\rangle\in\mathfrak{H},$$
 ${\mathsf{a}}^\dagger $ is the adjoint of ${\mathsf{a}} $ and viceversa. Now if $\mathfrak{H}=\overline{\text{right-span-over}~\quat}~\mathcal{O}_{\text{ar}}$, then it is a subspace of $L^2_{\quat}(\quat,d\varsigma(r, \theta,\phi, \psi))$ and
$$A_{f}:\mathfrak{H}\longrightarrow \mathfrak{H}~~\mbox{~~by}\quad
A_{f}(u)=A_{f}\mid u\rangle=\int_{\quat}\mid\gamma_{\qu}\rangle f(\qu,\overline{\qu})\langle\gamma_{\qu}\mid u\rangle d\varsigma(r, \theta,\phi, \psi),$$ for all $u\in \mathfrak{H}$.
Moreover, for each $u\in\mathfrak{H},~~A_{f}\mid u\rangle\in\mathfrak{H}$. For $|u\rangle, |v\rangle\in\mathfrak{H}$, it can also be considered as a function
$$A_{f}:\mathfrak{H}\times \mathfrak{H}\longrightarrow \quat~~\mbox{~~by}\quad
A_{f}(u,v)=\langle u\mid A_{f}\mid v\rangle=\int_{\quat}\langle u\mid\gamma_{\qu}\rangle f(\qu,\overline{\qu})\langle\gamma_{\qu}\mid v\rangle d\varsigma(r, \theta,\phi, \psi).$$
Since $\mid\gamma_{\qu}\rangle$ is a column vector and $\langle\gamma_{\qu}\mid$ is a row vector, we can see that the operator $A_{f}$ is a matrix and the matrix elements with respect to the basis $\{\mid e_{n}\rangle\}$ are given by
$$(A_{f})_{mn}=\langle e_{m}\mid A_{f}\mid e_{n}\rangle=\int_{\quat}\langle e_{m}\mid\gamma_{\qu}\rangle f(\qu,\overline{\qu})\langle\gamma_{\qu}\mid e_{n}\rangle d\varsigma(r, \theta,\phi, \psi).$$
Since
$$\langle e_{m}\mid\gamma_{\qu}\rangle=\mathcal N (\mid\qu\mid)^{-\frac{1}{2}}~\overline{\phi_{m}(\qu)}$$ and$$\langle\gamma_{\qu}\mid e_{n}\rangle=\overline{\langle e_{n}\mid\gamma_{\qu}\rangle}=\mathcal N (\mid\qu\mid)^{-\frac{1}{2}}~{\phi_{n}(\qu)},$$
we have
$$(A_{f})_{mn}=\int_{\quat}\mathcal N (\mid\qu\mid)^{-1}\overline{\phi_{m}(\qu)}f(\qu,\overline{\qu})\phi_{n}(\qu). d\varsigma(r, \theta,\phi, \psi).$$
Hence, it can easily be seen that
\begin{eqnarray*}
	({\mathsf{a}} )_{l,m}&=&\langle e_l|{\mathsf{a}} |e_m\rangle=\left\{\begin{array}{ccc}
		\sqrt{l+1}&\text{if}&m=l+1\\
		0&\text{if}&m\not=l+1,\end{array}\right.\\
	({\mathsf{a}}^\dagger )_{l,m}&=&\langle e_l|{\mathsf{a}}^\dagger |e_m\rangle=\left\{\begin{array}{ccc}
		\sqrt{l}&\text{if}&m=l-1\\
		0&\text{if}&m\not=l-1.\end{array}\right.
\end{eqnarray*}
\noindent
Let us realize the operator $A_f$ as annihilation and creation operators. From (\ref{Aq}) and (\ref{Aqb}) we have ${\mathsf{a}} \mid e_{0}\rangle=0\,$,
$${\mathsf{a}} \mid e_{m}\rangle=\sqrt{m}\mid e_{m-1}\rangle\,; ~ m=1,2,\cdots$$ and
$${\mathsf{a}}^\dagger \mid e_{m}\rangle=\sqrt{m+1}\mid e_{m+1}\rangle\,; ~ m=0,1,2,\cdots$$
That is, ${\mathsf{a}} ,{\mathsf{a}}^\dagger $ are annihilation and creation operators respectively. Moreover, one can easily see that ${\mathsf{a}} \mid\gamma_{\qu}\rangle=\mid\gamma_{\qu}\rangle \qu$, which is in complete analogy with the action of the annihilation operator on the ordinary harmonic oscillator CS and the result obtained in \cite{Thi2}. We can also write
$$\mid e_n\rangle=\frac{({\mathsf{a}}^\dagger )^n}{\sqrt{n!}}\mid e_0\rangle.$$
Further, real numbers commute with quaternions. According to (\ref{lft_mul-op}) let us compute
\begin{eqnarray*}
	(\qu \cdot {\mathsf{a}}^\dagger )^2\mid e_0\rangle
	&=&(\qu\cdot {\mathsf{a}}^\dagger )(\qu\cdot {\mathsf{a}}^\dagger )|e_0\rangle\\
	&=&(\qu\cdot {\mathsf{a}}^\dagger )(\qu\cdot {\mathsf{a}}^\dagger |e_0\rangle)\\
	&=&(\qu\cdot {\mathsf{a}}^\dagger )(\qu\cdot|e_1\rangle)\sqrt{1}\\
	&=&(\qu\cdot {\mathsf{a}}^\dagger )|e_1\rangle \qu\\
	&=&\qu\cdot( {\mathsf{a}}^\dagger |e_1\rangle) \qu\\
	&=&\qu\cdot(|e_2\rangle\sqrt{2}) \qu\\
	&=&|e_2\rangle\qu^2\sqrt{2!}.
\end{eqnarray*}
That is, $\dfrac{(\qu \cdot {\mathsf{a}}^\dagger )^2\mid e_0\rangle}{\sqrt{2!}}=|e_2\rangle\qu^2$. By induction, for each $n=0,1,2,\cdots$, we have $\dfrac{(\qu \cdot {\mathsf{a}}^\dagger )^n\mid e_0\rangle}{\sqrt{n!}}=|e_n\rangle\qu^n$.
Using this, one can see that
\begin{eqnarray*}
	(e^{-|\qu|^2/2}\cdot e^{\qu\cdot {\mathsf{a}}^\dagger })|e_0\rangle
	&=&e^{-|\qu|^2/2}\cdot(e^{\qu\cdot {\mathsf{a}}^\dagger }|e_0\rangle)\\
	&=&e^{-|\qu|^2/2}\cdot\left[ \sum_{n=0}^\infty\dfrac{(\qu \cdot {\mathsf{a}}^\dagger )^n\mid e_0\rangle}{n!}\right] \\
	&=&e^{-|\qu|^2/2}\cdot\left[ \sum_{n=0}^\infty|e_n\rangle\dfrac{\qu^n}{\sqrt{n!}}\right] \\
	&=&\mid\gamma_{\qu}\rangle.
\end{eqnarray*}
That is, $\mid\gamma_{\qu}\rangle=(e^{-|\qu|^2/2}\cdot e^{\qu\cdot {\mathsf{a}}^\dagger })|e_0\rangle$.
For quaternionic exponentials we refer the reader to \cite{Eb} (pp 204).
Now a direct calculation shows that
\begin{eqnarray*}
	{\mathsf{a}} {\mathsf{a}}^\dagger
	&=&\sum_{m=0}^{\infty}(m+1)\mid e_{m}\rangle\langle e_{m}\mid
\end{eqnarray*}
and
\begin{eqnarray*}
	{\mathsf{a}}^\dagger {\mathsf{a}}
	&=&\sum_{m=0}^{\infty}(m+1)\mid e_{m+1}\rangle\langle e_{m+1}\mid.
\end{eqnarray*}
Therefore the commutator of ${\mathsf{a}}^\dagger ,{\mathsf{a}} $ takes the form
\begin{eqnarray*}
	[{\mathsf{a}} ,{\mathsf{a}}^\dagger ]&=&{\mathsf{a}} {\mathsf{a}}^\dagger -{\mathsf{a}}^\dagger {\mathsf{a}} \\
	&=&\sum_{m=0}^{\infty}\mid e_{m}\rangle\langle e_{m}\mid
	=\mathbb{I}_{\mathfrak{H}}.
\end{eqnarray*}
\begin{remark}
The operator $A_f$ in (\ref{Qmap}) is formed by the vector $\mid\gamma_{\qu}\rangle f(\qu,\overline{\qu})$, which is the right scalar multiple of the vector $\mid\gamma_{\qu}\rangle$ by the scalar $ f(\qu,\overline{\qu})$, and the dual vector $\langle\gamma_{\qu}\mid$. Instead if one takes
\begin{equation}\label{Af}
A_f=\int_\quat f(\qu,\overline{\qu})\mid\gamma_{\qu}\rangle\langle\gamma_{\qu}\mid d\varsigma(r, \theta,\phi, \psi),
\end{equation}
then it is formed by $ f(\qu,\overline{\qu})\mid\gamma_{\qu}\rangle$ (a left scalar multiple of a right Hilbert space vector) and the dual vector $ \langle\gamma_{\qu}\mid$, which is  unconventional . Further, due to the noncommutativity of quaternions, an $A_f$ in the form (\ref{Af}) would have caused severe technical problems in the computations done in the sequel.
	
However, using the left multiplication defined on the right quaternion Hilbert space, we can write
\begin{equation}
\label{ca}
A_f=\int_{\quat}f(\qu, \oqu)\cdot |\gamma_{\qu}\rangle\langle\gamma_{\qu}|d\xi(r,\theta,\phi,\psi).
\end{equation}
In this case, since CS are linear combinations of the basis vectors by Proposition \ref{lft_mul} (f) the $A_f$ in Equations \ref{ca} and \ref{Qmap} become identical.
\end{remark}
The following Proposition demonstrate commutativity between quaternions and the right linear operators $\as$ and $\asd$. Further, it plays an important role in defining the momentum operator and hence in the following theory. For this reason, we give a detailed proof of the proposition.
\begin{proposition}\label{xAq}
	For each $\mathfrak{q}\in\quat$, we have $\mathfrak{q}\cdot {\mathsf{a}} ={\mathsf{a}} \cdot\mathfrak{q}$ and $\mathfrak{q}\cdot {{\mathsf{a}}^\dagger} ={{\mathsf{a}}^\dagger} \cdot\mathfrak{q}$.
\end{proposition}
\begin{proof}
	Let $\mathfrak{q}\in\quat$, and $\phi\in\mathfrak{H}$, then
	\begin{eqnarray*}
		(\mathfrak{q}\cdot {\mathsf{a}} )\phi&=&\sum_{n=0}^{\infty}|e_n\rangle\mathfrak{q}\langle e_n|{\mathsf{a}} \phi\rangle\\
		&=&\sum_{n=0}^{\infty}|e_n\rangle\mathfrak{q}\left(\sum_{m=0}^{\infty}\sqrt{m+1} \langle e_n|e_m\rangle\langle e_{m+1}|\phi\rangle\right)\\
		&=& \sum_{n=0}^{\infty}\sqrt{n+1}\, |e_n\rangle\mathfrak{q}\langle e_{n+1}|\phi\rangle.
	\end{eqnarray*}
For any arbitrary $\psi\in\mathfrak{H}$, we have
			$$\langle\psi\mid\mathfrak{q}\cdot\phi\rangle=\langle\psi\mid\sum_{m=0}^{\infty} e_m\,\mathfrak{q}\langle e_m|\phi\rangle\rangle=\sum_{m=0}^{\infty}\langle\psi\mid e_m\,\mathfrak{q}\langle e_m|\phi\rangle\rangle$$ hence
			$$\langle\psi\mid\mathfrak{q}\cdot\phi\rangle=\sum_{m=0}^{\infty}\langle\psi\mid e_m\rangle\mathfrak{q}\langle e_m|\phi\rangle.$$
			That is, $$\mid\mathfrak{q}\cdot\phi\rangle=\sum_{m=0}^{\infty}\mid e_m\rangle\mathfrak{q}\langle e_m|\phi\rangle.$$ Thus, by the orthogonality, it becomes, $\displaystyle\langle e_{n+1}|\mathfrak{q}\cdot\phi\rangle=\sum_{m=0}^{\infty} \langle e_{n+1}|e_m\rangle\mathfrak{q}\langle e_m|\phi\rangle=\mathfrak{q}\langle e_{n+1}|\phi\rangle$. Therefore,
	\begin{eqnarray*}
		({\mathsf{a}} \cdot\mathfrak{q} )\phi&=&{\mathsf{a}} (\mathfrak{q}\cdot\phi)\\
		&=&	\sum_{n=0}^{\infty}\sqrt{n+1} |e_n\rangle\langle e_{n+1}|\mathfrak{q}\cdot\phi\rangle\\
		&=&	\sum_{n=0}^{\infty}\sqrt{n+1} |e_n\rangle\left(\sum_{m=0}^{\infty} \langle e_{n+1}|e_m\rangle\mathfrak{q}\langle e_m|\phi\rangle\right) \\
		&=& \sum_{n=0}^{\infty}\sqrt{n+1}\, |e_n\rangle\mathfrak{q}\langle e_{n+1}|\phi\rangle.
	\end{eqnarray*}
	That is, $(\mathfrak{q}\cdot {\mathsf{a}} )\phi=({\mathsf{a}} \cdot\mathfrak{q} )\phi$. Since $\phi\in\mathfrak{H}$ is arbitrary, we have $\mathfrak{q}\cdot {\mathsf{a}} ={\mathsf{a}} \cdot\mathfrak{q} $. Similarly $\mathfrak{q}\cdot {{\mathsf{a}}^\dagger} ={{\mathsf{a}}^\dagger} \cdot\mathfrak{q}$ can be obtained.
\end{proof}
\subsection{Number, position and momentum operators and Hamiltonian}
Let $N={\mathsf{a}}^\dagger {\mathsf{a}} $, then we have
\begin{eqnarray*}
	N\mid e_n\rangle&=&{\mathsf{a}}^\dagger {\mathsf{a}} \mid e_n\rangle\\
	&=&\sum_{m=0}^{\infty}\mid e_{m+1}\rangle\langle e_{m+1}\mid e_n\rangle(m+1)\\
	&=&\mid e_n\rangle n.
\end{eqnarray*}
Therefore $N$ acts as the number operator and the Hilbert space $\mathfrak{H}$ is the quaternionic Fock space (for the quaternionic Fock spaces see \cite{Al}). As an analogue of the usual harmonic oscillator Hamiltonian, if we take $\mathcal{H}_h=N+\mathbb{I}_{\mathfrak{H}}$, then  $\mathcal{H}_h\mid e_n\rangle=\mid e_n\rangle (n+1)$, which is a Hamiltonian in the right quaternionic Hilbert space $\mathfrak{H}$ with spectrum $(n+1)$ and eigenvector $\mid e_n\rangle$.
Following the complex formalism, for $\qu\in \quat$ if we take $\displaystyle\mathit{q}=\frac{1}{\sqrt{2}}(\qu+\overline{\qu})$, then we can have a linear self-adjoint position operator as $$\displaystyle Q=\frac{1}{\sqrt{2}}({\mathsf{a}} +{\mathsf{a}}^\dagger ).$$
We note that $Q$ is self-adjoint
since the operators ${\mathsf{a}}$, and ${\mathsf{a}}^\dagger$ are defined on the whole Hilbert space and are one the adjoint of the other. As it was indicated in the introduction, under the right multiplication on a right quaternionic Hilbert space a linear self-adjoint momentum operator cannot be obtained by mimicking the complex momentum coordinate with one of $i, j$ or $k$ of the quaternionic units. For example, if we take $p=-\frac{i}{\sqrt{2}}(\qu-\oqu)$ then the momentum operator
$P=-\frac{i}{\sqrt{2}}(\as-\asd)$ becomes non self-adjoint. This point is very well discussed in \cite{MuTh}.

Now let us turn our attention to the momentum operators with a left multiplication defined on a right quaternionic Hilbert space. In fact, we show that there is a class of linear and self-adjoint momentum operators all resulting in the same harmonic oscillator Hamiltonian. For this purpose, we first take
$$\mathit{p}_i=\frac{-i}{\sqrt{2}}(\qu-\overline{\qu}),$$
$$\mathit{p}_j=\frac{-j}{\sqrt{2}}(\qu-\overline{\qu})$$
and
$$\mathit{p}_k=\frac{-k}{\sqrt{2}}(\qu-\overline{\qu});$$
then the momentum operators with respect to the above coordinates become
$$P_i=\frac{-i}{\sqrt{2}}\cdot({\mathsf{a}} -{\mathsf{a}}^\dagger ),$$
$$P_j=\frac{-j}{\sqrt{2}}\cdot({\mathsf{a}} -{\mathsf{a}}^\dagger )$$
and
$$P_k=\frac{-k}{\sqrt{2}}\cdot({\mathsf{a}} -{\mathsf{a}}^\dagger )$$ respectively. Note that the above operators $P_\tau:\,\tau=i,j,k$ has no issue on the domain, since the operators ${\mathsf{a}}$ and ${\mathsf{a}}^\dagger$ are defined on the whole Hilbert space, moreover the operators $P_\tau$ are self-adjoint for any $\tau\in\{i,j,k\}$. In fact, as we already pointed out,  ${\mathsf{a}}^\dagger $ is the adjoint of ${\mathsf{a}} $ and viceversa, so for any $\tau\in\{i,j,k\}$ we have:
\begin{eqnarray*}
	P_\tau^\dagger&=&\left[\frac{-\tau}{\sqrt{2}}\cdot({\mathsf{a}} -{\mathsf{a}}^\dagger )\right]^\dagger\\
	&=&({\mathsf{a}} ^\dagger-{{\mathsf{a}}^\dagger} ^\dagger)\cdot\frac{~\tau}{\sqrt{2}}~~\mbox{~~~~by (\ref{sc_mul_aj-op})}\\
	&=& ({\mathsf{a}}^\dagger -{\mathsf{a}} )\cdot\frac{~\tau}{\sqrt{2}}\\
	&=&\frac{-\tau}{\sqrt{2}}\cdot({\mathsf{a}} -{\mathsf{a}}^\dagger )~~\mbox{~~~~by Proposition \ref{xAq}}\\
	&=&P_\tau.
\end{eqnarray*}
Let us see the Hamiltonian with our position and momentum coordinates. For any $\tau\in\{i,j,k\}$ we have $H_\tau=\frac{1}{2}\left(\mid\mathit{q}\mid^{2}+\mid\mathit{p}_\tau\mid^{2}\right)=|\qu|^2$. Even if we use the position coordinate and all three momentum coordinates, we get
$$H_c=\frac{1}{2}\left(\mathit{q}^{2}-\mathit{p}_i^{2}-\mathit{p}_j^{2}-\mathit{p}_k^{2}\right)=|\qu|^2.$$

The lower symbol of $N$ is $\langle \gamma_{\qu}\mid N\mid\gamma_{\qu}\rangle=|\qu|^2$ and through a rather lengthy calculation, with the quantization map (\ref{Qmap}), we can see that  $A_{|\qu|^2}=N+\mathbb{I}_{\mathfrak{H}}$.
Now, let us turn our attention to the commutator and canonical quantization of the Hamiltonian. For each $\tau\in\{i,j,k\}$, we can calculate
\begin{eqnarray*}
	QP_\tau\phi
	&=&\left[\frac{({\mathsf{a}} +{\mathsf{a}}^\dagger )}{\sqrt{2}}\right]\left[(-\tau)\cdot\frac{({\mathsf{a}}
		-{\mathsf{a}}^\dagger )}{\sqrt{2}}\right]\phi\\
	&=&\left[\frac{({\mathsf{a}} +{\mathsf{a}}^\dagger )}{\sqrt{2}}\right]\left[(-\tau)\cdot\left( \frac{({\mathsf{a}}
		-{\mathsf{a}}^\dagger )}{\sqrt{2}}\phi\right) \right]\\
	&=&\left[\frac{({\mathsf{a}} +{\mathsf{a}}^\dagger )}{\sqrt{2}}\cdot(-\tau)\right]\left[\left( \frac{({\mathsf{a}}
		-{\mathsf{a}}^\dagger )}{\sqrt{2}}\phi\right) \right]~~~\mbox{~by (\ref{rgt_mul-op})}\\
	&=&\left[(-\tau)\cdot\left( \frac{({\mathsf{a}} +{\mathsf{a}}^\dagger )}{\sqrt{2}}\right) \right]\left[\left( \frac{({\mathsf{a}}
		-{\mathsf{a}}^\dagger )}{\sqrt{2}}\phi\right) \right]~~~\mbox{~by Proposition \ref{xAq}}\\
	&=&-\frac{1}{2}\tau\cdot\,[{{\mathsf{a}} }^{2}+{\mathsf{a}}^\dagger {\mathsf{a}} -{\mathsf{a}} {\mathsf{a}}^\dagger -{{\mathsf{a}}^\dagger }^{2}]\phi
\end{eqnarray*}
and
\begin{eqnarray*}
	P_\tau Q\phi
	&=&\left[-\tau\cdot\frac{({\mathsf{a}} -{\mathsf{a}}^\dagger )}{\sqrt{2}}\right]
	\left[\frac{({\mathsf{a}} +{\mathsf{a}}^\dagger )}{\sqrt{2}}\right]\phi\\
	&=&-\frac{1}{2}\tau\cdot\,[{{\mathsf{a}} }^{2}-{\mathsf{a}}^\dagger {\mathsf{a}} +{\mathsf{a}} {\mathsf{a}}^\dagger -{{\mathsf{a}}^\dagger }^{2}]\phi,
\end{eqnarray*}
for all $\phi\in V_\quat^R$. Therefore, for each $\tau\in\{i,j,k\}$, we have the commutator  $$[Q,P_\tau]=QP_\tau-P_\tau Q=\tau\cdot\,[{\mathsf{a}} ,{\mathsf{a}}^\dagger ]=\tau\cdot\mathbb{I}_{\mathfrak{H}}.$$
We can also obtain, in a similar fashion, for each $\tau\in\{i,j,k\}$,
\begin{eqnarray*}
	Q^{2}&=&
	\frac{1}{2}\,[{{\mathsf{a}} }^{2}+{\mathsf{a}}^\dagger {\mathsf{a}} +{\mathsf{a}} {\mathsf{a}}^\dagger
	+{{\mathsf{a}}^\dagger }^{2}]\quad\text{and}\\
	P^{2}_\tau&=&
	-\frac{1}{2}\,[{{\mathsf{a}} }^{2}-{\mathsf{a}}^\dagger {\mathsf{a}} -{\mathsf{a}} {\mathsf{a}}^\dagger
	+{{\mathsf{a}}^\dagger }^{2}].
\end{eqnarray*}
Hence, for each $\tau\in\{i,j,k\}$,
\begin{eqnarray*}
	\hat{H}_\tau&=&\frac{Q^{2}+P^{2}_\tau}{2}=\frac{1}{2}[{\mathsf{a}}^\dagger {\mathsf{a}} +{\mathsf{a}} {\mathsf{a}}^\dagger ]\\
	&=&{\mathsf{a}}^\dagger {\mathsf{a}} +\frac{1}{2}[{\mathsf{a}} {\mathsf{a}}^\dagger -{\mathsf{a}}^\dagger {\mathsf{a}} ]\\
	&=&N+\frac{1}{2}\mathbb{I}_{\mathfrak{H}},
\end{eqnarray*}
which does not depend on the choice of $\tau\in\{i,j,k\}$. Let us consider the momentum coordinate
$$\mathit{p}^*=-\frac{(i+j+k)}{\sqrt{3}}\frac{(\qu-\overline{\qu})}{\sqrt{2}}$$
to define another momentum operator $P^*$ as  $$P^*=-\frac{(i+j+k)}{\sqrt{3}}\cdot\frac{({\mathsf{a}} -{\mathsf{a}}^\dagger )}{\sqrt{2}}.$$
One can realize that $P^*$ is self-adjoint, and  the Hamiltonian becomes
$$H^*=\frac{1}{2}\left(\mid\mathit{q}\mid^{2}+\mid\mathit{p}^*\mid^{2}\right)=|\qu|^2.$$
Furthermore, we have $$[Q,P^*]=\frac{(i+j+k)}{\sqrt{3}}\cdot\mathbb{I}_{\mathfrak{H}}$$ and
$$\hat{H}^*=\frac{Q^{2}+{P^*}^{2}}{2}=N+\frac{1}{2}\mathbb{I}_{\mathfrak{H}}.$$
In a more wider range, we can define the momentum coordinate for each $I\in\mathbb{S}$, such that
$$\mathit{p}_I=\frac{-I}{\sqrt{2}}(\qu-\overline{\qu}),$$
and the corresponding momentum operator is
$$P_I=\frac{-I}{\sqrt{2}}\cdot({\mathsf{a}} -{\mathsf{a}}^\dagger ),$$
which is also self-adjoint.
In this set up, the coordinate for the Hamiltonian is $H_{I}=\frac{1}{2}\left(\mid\mathit{q}\mid^{2}+\mid\mathit{p}_I\mid^{2}\right)=|\qu|^2$. We also have
$$[Q,P_I]=I\cdot\mathbb{I}_{\mathfrak{H}}$$ and
\begin{equation}\label{xxx}
\hat{H}_I=\frac{Q^{2}+P^{2}_I}{2}=N+\frac{1}{2}\mathbb{I}_{\mathfrak{H}}.
\end{equation}
In conclusion, in the quaternionic case, we have a set of momentum operators:
$$\mathfrak{P}=\left\lbrace P_I=\frac{-I}{\sqrt{2}}\cdot({\mathsf{a}} -{\mathsf{a}}^\dagger )~\mid~I\in\mathbb{S}\right\rbrace. $$
\begin{remark}{\rm We point out that in various cases, for example to compute the Hamiltonian, one can choose a specific $I\in\mathbb S$ and to work with $P_I$. However, there is no preferred choice and there is no need to fix $I\in\mathbb S$. A similar phenomenon occurs with the elements of the spectrum which are always spheres. In various cases one may fix a specific element in a sphere and to work with it, but there is no necessity for such a choice. The quaternionic setting offers a larger number of possibilities, compared to the complex case. This fact was already observed by Adler in \cite{Ad}, section 2.3.}
\end{remark}
\subsection{Heisenberg uncertainty}
In this subsection, since we have good candidates for the position and momentum operators, let us demonstrate the Heisenberg uncertainty relation in the whole set of quaternions and on a quaternionic slice. In fact, we shall show that the RQCS become the minimum uncertainty states in a neighbourhood of the origin when we consider the whole set of quaternions and the set of RQCS becomes the minimum uncertainty states and the so-called intelligent states on a quaternion slice.\\

In order to compute the expectation values of the involved operators recall that
\begin{eqnarray*}
    {\mathsf{a}} |e_0\rangle&=&0\\
	{\mathsf{a}} |e_m\rangle&=&\sqrt{m}|e_{m-1}\rangle;\quad m=1,2,\cdots\\
	{\mathsf{a}}^\dagger|e_m\rangle&=&\sqrt{m+1}|e_{m+1}\rangle;\quad m=0,1,\cdots
\end{eqnarray*}
and
\begin{equation}\label{ev}
{\mathsf{a}} |\gamma_{\qu}\rangle=|\gamma_{\qu}\rangle\qu.
\end{equation}
Using (\ref{ev}) we can easily see that
$${\mathsf{a}} ^2|\gamma_{\qu}\rangle={\mathsf{a}} |\gamma_{\qu}\rangle\qu=|\gamma_{\qu}\rangle\qu^2.$$
Hence, as $\langle\gamma_{\qu}|\gamma_{\qu}\rangle=1$, we get
$$\langle\gamma_{\qu}|{\mathsf{a}} |\gamma_{\qu}\rangle=\qu\quad\mbox{and}\quad
\langle\gamma_{\qu}|{\mathsf{a}} ^2|\gamma_{\qu}\rangle=\qu^2.$$
For the sake of simplicity, we set $a_m=\sqrt{m+1}$ and $b_m=\sqrt{(m+1)(m+2)}$. The action of the operators, ${{\mathsf{a}}^\dagger} , {{\mathsf{a}}^\dagger} ^2, {{\mathsf{a}}^\dagger} {\mathsf{a}} $ and ${\mathsf{a}} {{\mathsf{a}}^\dagger} $ on the RQCS takes the form
\begin{eqnarray*}
	{{\mathsf{a}}^\dagger} |\gamma_{\qu}\rangle
	&=&e^{-|\qu|^2/2}\sum_{m=0}^{\infty}{{\mathsf{a}}^\dagger} |e_m\rangle\frac{\qu^m}{\sqrt{m!}}\\
	&=&e^{-|\qu|^2/2}\sum_{m=0}^{\infty}|e_{m+1}\rangle a_m\frac{\qu^m}{\sqrt{m!}},
\end{eqnarray*}
and similarly,
\begin{eqnarray*}
	{{\mathsf{a}}^\dagger} ^2|\gamma_{\qu}\rangle
	&=&e^{-|\qu|^2/2}\sum_{m=0}^{\infty}|e_{m+2}\rangle b_m\frac{\qu^m}{\sqrt{m!}},
\end{eqnarray*}
\begin{eqnarray*}
	{{\mathsf{a}}^\dagger} {\mathsf{a}} |\gamma_{\qu}\rangle
	&=&e^{-|\qu|^2/2}\sum_{m=0}^{\infty}|e_{m+1}\rangle a_m\frac{\qu^{m+1}}{\sqrt{m!}}
\end{eqnarray*}
and
\begin{eqnarray*}
	{\mathsf{a}} {{\mathsf{a}}^\dagger} |\gamma_{\qu}\rangle
	&=&e^{-|\qu|^2/2}\sum_{m=0}^{\infty}|e_{m}\rangle a_m^2\frac{\qu^m}{\sqrt{m!}}.
\end{eqnarray*}
The dual of the RQCS is
$$\langle\gamma_{\qu}|
=e^{-|\qu|^2/2}\sum_{m=0}^{\infty}\frac{\oqu^m}{\sqrt{m!}}\langle e_{m}|.$$
Thereby we get the expectation values
\begin{eqnarray*}
	\langle\gamma_{\qu}|{{\mathsf{a}}^\dagger} |\gamma_{\qu}\rangle
	&=&
	e^{-|\qu|^2}\sum_{m=0}^{\infty}\sum_{n=0}^{\infty}\frac{\oqu^m}{\sqrt{m!}}
	\langle e_m|e_{n+1}\rangle a_n\frac{\qu^n}{\sqrt{n!}}\\
	&=&e^{-|\qu|^2}\sum_{m=0}^{\infty}\frac{\oqu^{m+1}\qu^m}{m!}\\
	&=&e^{-|\qu|^2}\oqu\sum_{m=0}^{\infty}\frac{|\qu|^{2m}}{m!}\\
	&=&\oqu,
\end{eqnarray*}
and similarly,
\begin{eqnarray*}
	\langle\gamma_{\qu}|{{\mathsf{a}}^\dagger} ^2|\gamma_{\qu}\rangle&=&
	\oqu^2,\\
	\langle\gamma_{\qu}|{{\mathsf{a}}^\dagger} {\mathsf{a}} |\gamma_{\qu}\rangle
	&=&\oqu\qu=|\qu|^2,\\
	\langle\gamma_{\qu}|{\mathsf{a}} {{\mathsf{a}}^\dagger} |\gamma_{\qu}\rangle
	&=&1+|\qu|^2.
\end{eqnarray*}
Using the above expectation values we can get the expectation values of $Q$ and $ Q^2$ as follows.
\begin{eqnarray*}
	\langle\gamma_{\qu}|Q|\gamma_{\qu}\rangle&=&
	\frac{1}{\sqrt{2}}\langle\gamma_{\qu}|{\mathsf{a}} +{{\mathsf{a}}^\dagger} |\gamma_{\qu}\rangle\\
	&=&\frac{1}{\sqrt{2}}[\langle\gamma_{\qu}|{\mathsf{a}} |\gamma_{\qu}\rangle
	+\langle\gamma_{\qu}|{{\mathsf{a}}^\dagger} |\gamma_{\qu}\rangle]\\
	&=&\frac{1}{\sqrt{2}}(\qu+\oqu),
\end{eqnarray*}
and hence
$$\langle\gamma_{\qu}|Q|\gamma_{\qu}\rangle^2=\frac{1}{2}(\qu^2+2|\qu|^2+\oqu^2).$$
Now for $Q^2$
\begin{eqnarray*}
	\langle\gamma_{\qu}|Q^2|\gamma_{\qu}\rangle&=&
	\frac{1}{2}\langle\gamma_{\qu}|{\mathsf{a}} ^2+{\mathsf{a}} {\mathsf{a}}^\dagger +{\mathsf{a}}^\dagger {\mathsf{a}}
	{{\mathsf{a}}^\dagger} ^2|\gamma_{\qu}\rangle\\
	&=&\frac{1}{2}[\qu^2+1+|\qu|^2+|\qu|^2+\oqu^2]\\
	&=&\frac{1}{2}[\qu^2+1+2|\qu|^2+\oqu^2].
\end{eqnarray*}
Therefore the variance of $Q$ becomes
\begin{eqnarray*}
	\langle\Delta Q\rangle^2&=&\langle\gamma_{\qu}|Q^2|\gamma_{\qu}\rangle-\langle\gamma_{\qu}|Q|\gamma_{\qu}\rangle^2\\
	&=&1/2.
\end{eqnarray*}
That is,
$$\langle\Delta Q\rangle=\frac{1}{\sqrt{2}}.$$
Let $I\in\mathbb{S}$, then for the momentum operator $P_I$, we have
\begin{eqnarray*}
	P_I|\gamma_{\qu}\rangle&=&\left(\frac{-I}{\sqrt{2}}\cdot[{\mathsf{a}} -{{\mathsf{a}}^\dagger} ]\right)|\gamma_{\qu}\rangle\\
	&=&\frac{-I}{\sqrt{2}}\cdot\left([{\mathsf{a}} -{{\mathsf{a}}^\dagger} ]|\gamma_{\qu}\rangle\right)\\
	&=&\frac{-I}{\sqrt{2}}\cdot\left({\mathsf{a}} |\gamma_{\qu}\rangle-{{\mathsf{a}}^\dagger} |\gamma_{\qu}\rangle\right).
\end{eqnarray*}
Since
$$I\cdot \mathsf{a}=\sum_{m=0}^{\infty}\sqrt{(m+1)}\mid e_{m}\rangle I\langle e_{m+1}\mid,$$ we have
\begin{eqnarray*}
	(I\cdot \mathsf{a})\mid\gamma_{\qu}\rangle
	&=&e^{-|\qu|^2/2}\sum_{m=0}^{\infty}\sqrt{(m+1)}\mid e_m\rangle I\sum_{n=0}^{\infty}\langle  e_{m+1}
	\vert
	e_n\rangle\frac{\qu^n}{\sqrt{n!}}  \\
	&=&e^{-|\qu|^2/2}\sum_{m=0}^{\infty}\sqrt{(m+1)}\mid e_m\rangle I\frac{\qu^{m+1}}{\sqrt{m+1!}}
\end{eqnarray*}
and
\begin{eqnarray*}
	\langle\gamma_{\qu}\mid(I\cdot \mathsf{a})\mid\gamma_{\qu}\rangle
	&=&	e^{-|\qu|^2/2}\sum_{m=0}^{\infty}\sqrt{(m+1)}\langle\gamma_{\qu}\mid e_m\rangle I\frac{\qu^{m+1}}{\sqrt{m+1!}}\\
	&=&e^{-|\qu|^2}\sum_{m=0}^{\infty}\sqrt{(m+1)} \sum_{n=0}^{\infty}\frac{\oqu^n}{\sqrt{n!}}\langle e_{n}\vert e_m\rangle I\frac{\qu^{m+1}}{\sqrt{m+1!}}\\
	&=&e^{-|\qu|^2}\sum_{m=0}^{\infty}\sqrt{(m+1)} \frac{\oqu^m}{\sqrt{m!}} I\frac{\qu^{m+1}}{\sqrt{m+1!}}\\
	&=&\left( e^{-|\qu|^2}\sum_{m=0}^{\infty} \frac{\oqu^m I \qu^{m}}{m!}\right)\qu=\mathfrak{C}_I\qu; \\
\end{eqnarray*}	
where $\displaystyle\mathfrak{C}_I=e^{-|\qu|^2}\sum_{m=0}^{\infty} \frac{\oqu^m I \qu^{m}}{m!}$ and this series  absolutely converges to $1$. That is, $|\mathfrak{C}_I|\leq1$. It is nice to note that, $\overline{{\mathfrak{C}}_I}=-\mathfrak{C}_I$ and $|\mathfrak{C}_I|^2=-\mathfrak{C}_I^2$. From this, we can say that there exist $\mathcal{I}\in\mathbb{S}$ and $r\in[0,1]$ such that $\mathfrak{C}_I=r\mathcal{I}$. Also we  find
$$(I\cdot {{\mathsf{a}}^\dagger} )\mid\gamma_{\qu}\rangle=e^{-|\qu|^2/2}\sum_{m=0}^{\infty}\sqrt{(m+1)}\mid e_{m+1}\rangle I\frac{\qu^{m}}{\sqrt{m!}}$$
and
$$\langle\gamma_{\qu}\mid(I\cdot {{\mathsf{a}}^\dagger} )\mid\gamma_{\qu}\rangle=\oqu\left( e^{-|\qu|^2}\sum_{m=0}^{\infty} \frac{\oqu^m I \qu^{m}}{m!}\right)=\oqu\mathfrak{C}_I.$$
Therefore,
\begin{eqnarray*}
	\langle\gamma_{\qu}|P_I|\gamma_{\qu}\rangle&=&	\frac{1}{\sqrt{2}} [\langle\gamma_{\qu}|(I\cdot {\mathsf{a}} )|\gamma_{\qu}\rangle	-\langle\gamma_{\qu}|(I\cdot {{\mathsf{a}}^\dagger} )|\gamma_{\qu}\rangle]\\	&=&\frac{1}{\sqrt{2}}(\mathfrak{C}_I\qu-\oqu\mathfrak{C}_I).
\end{eqnarray*}
Hence we obtain
\begin{eqnarray*}
	\langle\gamma_{\qu}|P_I|\gamma_{\qu}\rangle^2
	&=&\frac{1}{2}(\mathfrak{C}_I\qu-\oqu\mathfrak{C}_I)^2\\
	&=&\frac{1}{2}(\mathfrak{C}_I\qu+\overline{\mathfrak{C}_I\qu})^2\\
	&=&\frac{1}{2}[(\mathfrak{C}_I\qu)^2+2|\mathfrak{C}_I\qu|^2+(\overline{\mathfrak{C}_I\qu})^2].
\end{eqnarray*}
Since $I^2=-1$, we have
\begin{eqnarray*}
	\langle\gamma_{\qu}|P_I^2|\gamma_{\qu}\rangle&=&
	-\frac{1}{2}\langle\gamma_{\qu}|{\mathsf{a}} ^2-{\mathsf{a}} {{\mathsf{a}}^\dagger} -{{\mathsf{a}}^\dagger} {\mathsf{a}}
	+{{\mathsf{a}}^\dagger} ^2|\gamma_{\qu}\rangle\\
	&=&-\frac{1}{2}[\qu^2-1-|\qu|^2-|\qu|^2+\oqu^2]\\
	&=&-\frac{1}{2}[\qu^2-1-2|\qu|^2+\oqu^2].
\end{eqnarray*}
Therefore the variance of $P_I$ becomes
\begin{eqnarray*}
	\langle\Delta P_I\rangle^2&=&\langle\gamma_{\qu}|P_I^2|\gamma_{\qu}\rangle-\langle\gamma_{\qu}|P_I|\gamma_{\qu}\rangle^2\\
	&=&-\frac{1}{2}[\qu^2-1-2|\qu|^2+\oqu^2]-\frac{1}{2}[(\mathfrak{C}_I\qu)^2+2|\mathfrak{C}_I\qu|^2
+(\overline{\mathfrak{C}_I\qu})^2],
\end{eqnarray*}
which is a nonnegative real number, since
Re$( \mathfrak{q}^2 - (\mathfrak{C}_I \mathfrak{q})^2) \leq 2(|\mathfrak{q}|^2+ |\mathfrak{C}_I \mathfrak{q}|^2) +1$.
Hence, we have
\begin{eqnarray*}
	\langle\Delta Q\rangle^2\langle\Delta P_I\rangle^2
	&=&-\frac{1}{4}[(\qu^2-1-2|\qu|^2+\oqu^2)+((\mathfrak{C}_I\qu)^2+2|\mathfrak{C}_I\qu|^2+(\overline{\mathfrak{C}_I\qu})^2)]\\
	&\geq&-\frac{1}{4}[(\qu^2-1-2|\qu|^2+\oqu^2)+((\mathfrak{C}_I\qu)^2+2|\qu|^2+(\overline{\mathfrak{C}_I\qu})^2)]~\mbox{~~as~~}|\mathfrak{C}_I|\leq1\\
	&=&\frac{1}{4}-\frac{1}{4}[(\qu^2+\oqu^2)+((\mathfrak{C}_I\qu)^2+(\overline{\mathfrak{C}_I\qu})^2)].\\
\end{eqnarray*}	
From this,
\begin{eqnarray*}
	|\langle\Delta Q\rangle^2\langle\Delta P_I\rangle^2|
	&\geq&\frac{1}{4}-\frac{1}{4}(|\qu|^2(1+|\mathfrak{C}_I|^2)+|\oqu|^2(1+|\mathfrak{C}_I|^2|))\\
	&\geq&\frac{1}{4}-\frac{1}{2}|\qu^2|\,(1+\mathfrak{C}_I|^2)\\
	&\geq&\frac{1}{4}-|\qu^2|~\mbox{~~as~~}|\mathfrak{C}_I|\leq1.
\end{eqnarray*}	
Thus $|\langle\Delta Q\rangle^2\langle\Delta P_I\rangle^2|\geq\frac{1}{4}-|\qu|^2$. Likewise, one can easily get  $|\langle\Delta Q\rangle^2\langle\Delta P_I\rangle^2|\leq\frac{1}{4}+|\qu|^2$. As a summary, we have $$|\,|\langle\Delta Q\rangle^2\langle\Delta P_I\rangle^2|-\frac{1}{4}|\leq|\qu|^2.$$
From this, one can say that $$\lim_{|\qu|\longrightarrow0}|\langle\Delta Q\rangle\langle\Delta P_I\rangle|=\frac{1}{2}.$$
Further, since $[Q,P_I]=I\cdot\mathbb{I}_{\mathfrak{H}}$, we have
\begin{eqnarray*}
	[Q,P_I]|\gamma_{\qu}\rangle&=&(I\cdot\mathbb{I}_{\mathfrak{H}})|\gamma_{\qu}\rangle
	=I\cdot(\mathbb{I}_{\mathfrak{H}}|\gamma_{\qu}\rangle)\\
	&=&|I\cdot\gamma_{\qu}\rangle.
\end{eqnarray*}
Therefore
$$\langle\gamma_{\qu}|[Q,P_I]|\gamma_{\qu}\rangle=\langle\gamma_{\qu}|I\cdot\gamma_{\qu}\rangle=e^{-|\qu|^2}\sum_{m=0}^{\infty} \frac{\oqu^m I \qu^{m}}{m!}=\mathfrak{C}_I=r\mathcal{I}.$$
Hence
$$\frac{1}{2}|\langle[Q,P_I]\rangle|=\frac{1}{2}|r\mathcal{I}|=\frac{1}{2}r\leq\frac{1}{2},$$
as $r\leq1$.
Therefore, we have
$$\lim_{|\qu|\longrightarrow0}|\langle\Delta Q\rangle\langle\Delta P_I\rangle|\geq\frac{1}{2}|\langle[Q,P_I]\rangle|.$$
The Heisenberg uncertainty  gets saturated only in a limit sense. We believe that this is not due to a defect in the definition of the momentum operator but it is just a technical issue in obtaining a closed form for the series $\displaystyle\sum_{m=0}^{\infty}\frac{\oqu^m I\qu^m}{m!}$.

In the above setup, the imaginary-unit element $I\in\mathbb{S}$ in the momentum operator $P_I$ was chosen arbitrarily, that is, there was no correlation between the quantization map $f:\qu\mapsto f(\qu,\bar{\qu})$ and the choice of imaginary-unit element $I\in\mathbb{S}$. But when there is a correlation between them, it is possible to have the Heisenberg uncertainty with saturation.

In order to see this, define a map $\mathfrak{U}:\quat\longrightarrow\mathbb{S}$ by $\mathfrak{U}(\qu)=I_\qu$, for all $\qu=x+I_\qu y\in\quat$. Note that on a quaternionic slice the $I_{\qu}$ will remain the same and therefore, in this set up, we are working on a quaternion slice. The map $\mathfrak{U}$ is well-defined and onto. Using this map, we shall define a momentum operator as
\begin{equation}\label{mom_op}
P=-\frac{\mathfrak{U}(\qu)}{\sqrt{2}}\cdot({\mathsf{a}} -{\mathsf{a}}^\dagger )=-\frac{I_\qu}{\sqrt{2}}\cdot({\mathsf{a}} -{\mathsf{a}}^\dagger )
\end{equation}
with the momentum coordinate
\begin{equation}\label{mom_cor}
\mathit{p}=-\frac{\mathfrak{U}(\qu)}{\sqrt{2}}(\qu-\overline{\qu})=-\frac{I_\qu  }{\sqrt{2}}(\qu-\overline{\qu}).
\end{equation}
Now, as before,	$P^\dagger=P$.
That is, the momentum operator $P$ is self-adjoint. The Hamiltonian $H$ is given by $H=\mathit{p}^2+\mathit{q}^2=|\qu|^2$. Furthermore,
\begin{eqnarray*}
	QP\phi
	&=&\left[\frac{({\mathsf{a}} +{\mathsf{a}}^\dagger )}{\sqrt{2}}\right]\left[(-I_\qu)\cdot\frac{({\mathsf{a}}
		-{\mathsf{a}}^\dagger )}{\sqrt{2}}\right]\phi
	=-\frac{1}{2}I_\qu\cdot\,[{{\mathsf{a}} }^{2}+{\mathsf{a}}^\dagger {\mathsf{a}} -{\mathsf{a}} {\mathsf{a}}^\dagger -{{\mathsf{a}}^\dagger }^{2}]\phi
\end{eqnarray*}
and
\begin{eqnarray*}
	P Q\phi
	&=&\left[-I_\qu\cdot\frac{({\mathsf{a}} -{\mathsf{a}}^\dagger )}{\sqrt{2}}\right]
	\left[\frac{({\mathsf{a}} +{\mathsf{a}}^\dagger )}{\sqrt{2}}\right]\phi
	=-\frac{1}{2}I_\qu\cdot\,[{{\mathsf{a}} }^{2}-{\mathsf{a}}^\dagger {\mathsf{a}} +{\mathsf{a}} {\mathsf{a}}^\dagger -{{\mathsf{a}}^\dagger }^{2}]\phi,
\end{eqnarray*}
for all $\phi\in \mathfrak{H}$. Therefore, we have the commutator $$[Q,P]=QP-PQ=I_\qu\cdot\,[{\mathsf{a}} ,{\mathsf{a}}^\dagger ]=I_\qu\cdot\mathbb{I}_{\mathfrak{H}}.$$
We can also obtain, as before,
\begin{eqnarray*}
	Q^{2}&=&
	\frac{1}{2}\,[{{\mathsf{a}} }^{2}+{\mathsf{a}}^\dagger {\mathsf{a}} +{\mathsf{a}} {\mathsf{a}}^\dagger
	+{{\mathsf{a}}^\dagger }^{2}]\quad\text{and}\\
	P^{2}&=&
	-\frac{1}{2}\,[{{\mathsf{a}} }^{2}-{\mathsf{a}}^\dagger {\mathsf{a}} -{\mathsf{a}} {\mathsf{a}}^\dagger
	+{{\mathsf{a}}^\dagger }^{2}].
\end{eqnarray*}
Hence
\begin{eqnarray*}
	\hat{H}&=&\frac{Q^{2}+P^{2}}{2}=\frac{1}{2}[{\mathsf{a}}^\dagger {\mathsf{a}} +{\mathsf{a}} {\mathsf{a}}^\dagger ]
	={\mathsf{a}}^\dagger {\mathsf{a}} +\frac{1}{2}[{\mathsf{a}} {\mathsf{a}}^\dagger -{\mathsf{a}}^\dagger {\mathsf{a}} ]
	=N+\frac{1}{2}\mathbb{I}_{\mathfrak{H}}.
\end{eqnarray*}
Furthermore we have
\begin{eqnarray*}
	P|\gamma_{\qu}\rangle
	&=&\frac{-I_\qu}{\sqrt{2}}\cdot\left({\mathsf{a}} |\gamma_{\qu}\rangle-{{\mathsf{a}}^\dagger}. |\gamma_{\qu}\rangle\right)
\end{eqnarray*}
and
$$I_\qu\cdot \mathsf{a}=\sum_{m=0}^{\infty}\sqrt{(m+1)}\mid e_{m}\rangle I_\qu\langle e_{m+1}\mid.$$ Thus
\begin{eqnarray*}
	(I_\qu\cdot \mathsf{a})\mid\gamma_{\qu}\rangle
	&=&e^{-|\qu|^2/2}\sum_{m=0}^{\infty}\sqrt{(m+1)}\mid e_m\rangle I_\qu\sum_{n=0}^{\infty}\langle  e_{m+1}
	\vert
	e_n\rangle\frac{\qu^n}{\sqrt{n!}}  \\
	&=&e^{-|\qu|^2/2}\sum_{m=0}^{\infty}\sqrt{(m+1)}\mid e_m\rangle I_\qu\frac{\qu^{m+1}}{\sqrt{m+1!}}.
\end{eqnarray*}
Since $\qu, \oqu$ and $I_{\qu}$ commute, we have
\begin{equation}\label{H1}
	\langle\gamma_{\qu}\mid(I_\qu\cdot \mathsf{a})\mid\gamma_{\qu}\rangle
	=\left( e^{-|\qu|^2}\sum_{m=0}^{\infty} \frac{\oqu^m  \qu^{m}}{m!}\right)I_\qu\qu=I_\qu\qu~\mbox{~~as~~}\sum_{m=0}^{\infty} \frac{\oqu^m \qu^{m}}{m!}=e^{|\qu|^2}
\end{equation}	
and
\begin{equation}\label{H2}
\langle\gamma_{\qu}\mid(I_\qu\cdot {{\mathsf{a}}^\dagger} )\mid\gamma_{\qu}\rangle=\oqu I_\qu\left( e^{-|\qu|^2}\sum_{m=0}^{\infty} \frac{\oqu^m  \qu^{m}}{m!}\right)=\oqu I_\qu .
\end{equation}
Hence,
\begin{eqnarray*}
	\langle\gamma_{\qu}|P|\gamma_{\qu}\rangle&=&	\frac{1}{\sqrt{2}} [\langle\gamma_{\qu}|(I_\qu\cdot {\mathsf{a}} )|\gamma_{\qu}\rangle	-\langle\gamma_{\qu}|(I_\qu\cdot {{\mathsf{a}}^\dagger} )|\gamma_{\qu}\rangle]=\frac{I_\qu}{\sqrt{2}}( \qu-\oqu )
\end{eqnarray*}
and
\begin{eqnarray*}
	\langle\gamma_{\qu}|P|\gamma_{\qu}\rangle^2
	&=&-\frac{1}{2}( \qu-\oqu )^2
	=-\frac{1}{2}[( \qu)^2-2| \qu|^2+(\overline{ \qu})^2].
\end{eqnarray*}
Since $I_\qu^2=-1$, we have
\begin{eqnarray*}
	\langle\gamma_{\qu}|P^2|\gamma_{\qu}\rangle&=&-\frac{1}{2}[\qu^2-1-2|\qu|^2+\oqu^2].
\end{eqnarray*}
Therefore the variance of $P$ becomes
\begin{eqnarray*}
	\langle\Delta P\rangle^2&=&\langle\gamma_{\qu}|P^2|\gamma_{\qu}\rangle-\langle\gamma_{\qu}|P|\gamma_{\qu}\rangle^2\\
	&=&-\frac{1}{2}[\qu^2-1-2|\qu|^2+\oqu^2]+\frac{1}{2}[\qu^2-2| \qu|^2+\overline{ \qu}^2]\\
	&=&\frac{1}{2}.
\end{eqnarray*}
Hence
$$\langle\Delta Q\rangle^2\langle\Delta P\rangle^2=\frac{1}{4}.$$
Further, since $[Q,P]=I_\qu\cdot\mathbb{I}_{\mathfrak{H}}$, we obtain
\begin{eqnarray*}
	[Q,P]|\gamma_{\qu}\rangle&=&(I_\qu\cdot\mathbb{I}_{\mathfrak{H}})|\gamma_{\qu}\rangle
	=I_\qu\cdot(\mathbb{I}_{\mathfrak{H}}|\gamma_{\qu}\rangle)\\
	&=&|I_\qu\cdot\gamma_{\qu}\rangle.
\end{eqnarray*}
Therefore
\begin{equation}\label{H3}
\langle\gamma_{\qu}|[Q,P]|\gamma_{\qu}\rangle=\langle\gamma_{\qu}|I_\qu\cdot\gamma_{\qu}\rangle=e^{-|\qu|^2}\sum_{m=0}^{\infty} \frac{\oqu^m I_\qu \qu^{m}}{m!}=I_\qu
\end{equation}
and
$$\frac{1}{2}|\langle[Q,P]\rangle|=\frac{1}{2}|I_\qu|=\frac{1}{2}.$$
The above can be recapitulated in one line as
$$\langle\Delta Q\rangle\langle\Delta P\rangle=\frac{1}{2}|\langle[Q,P]\rangle|=\frac{1}{2}.$$
That is, the RQCS, $|\gamma_{\qu}\rangle$, saturate the Heisenberg uncertainty on a quaternion slice and, due to  $[Q,P]=I_\qu\cdot\mathbb{I}_{\mathfrak{H}}$, the RQCS are minimum uncertainty states. States for which equality in the Heisenberg uncertainty relation is achieved are called intelligent states. In this sense the QRCS are intelligent states too, which is in complete analogy with the canonical CS of the complex quantum mechanics.
\begin{remark}
From the equalities (\ref{H1}), (\ref{H2}) and (\ref{H3}), it is evident that in the general case (the case of whole set of quaternions) what have prevented us in getting the Heisenberg uncertainty relation saturated is just the closed form of the series $\displaystyle\sum_{m=0}^{\infty}\frac{\oqu^m I\qu^m}{m!}$ but not the way the momentum operator is defined.
\end{remark}

\section{Some Algebraic Structures}
In this section, we shall investigate the Weyl-Heisenberg Lie Algebra and group
representation in the quaternionic setting.\\
 First of all, as in the complex quantum mechanics, all the operators considered here are unbounded operators. However, the operators act as $\HI\ni|\phi\rangle\mapsto |\psi\rangle\in\HI$, that is, the domain and the range of the operators are dense subsets of $\HI$. Furthermore, the Hilbert space, $\HI$, can be taken as a space right-spanned by the regular functions $\{\frac{\qu^m}{m!}~~|~~m\in\mathbb{N}\}$ or anti-regular functions $\{\frac{\oqu^m}{m!}~~|~~m\in\mathbb{N}\}$ over $\quat$ (counterparts of holomorphic and anti-holomorphic functions). In this respect, the operators considered here do not have any domain problems as for the operators in the complex quantum mechanics. Therefore, we can use the operator tools of complex quantum mechanics, in particular, the Baker-Campbell-Hausdorff formula (for a complex argument along these lines see chapter 14 in \cite{Brian}).

\subsection{Some Quaternionic Lie Algebras}
Let $\tau\in\{i,j,k\}$ and define
$$\mathfrak{A}_{\mathbb{C}_\tau}=\mbox{linear span over }\mathbb{C}_\tau\,\{\mathbb{I}_{\mathfrak{H}}, \mathsf{a}, \mathsf{a}^\dagger\};$$
where $\mathbb{C}_\tau=\{x=x_1+\tau x_2~|~x_1, x_2\in\mathbb{R}\}$. Then $\mathfrak{A}_{\mathbb{C}_\tau}$ is a vector space over $\mathbb{C}_\tau$. Define $$[\cdot,\cdot]_\tau:\mathfrak{A}_{\mathbb{C}_\tau}\times\mathfrak{A}_{\mathbb{C}_\tau}
\longrightarrow\mathfrak{A}_{\mathbb{C}_\tau} \quad{\text{by}}\quad
[\mathcal{A}, \mathcal{B}]_\tau=\mathcal{A} \mathcal{B}-\mathcal{B} \mathcal{A},\mbox{~~for all~~}\mathcal{A}, \mathcal{B}\in\mathfrak{A}_{\mathbb{C}_\tau}.$$
One can easily see that the bracket $[\cdot, \cdot]$ satisfies the following axioms:
\begin{itemize}
	\item[(a)] \textit{Bilinearity}: for all $x,y\in\mathbb{R}$ and $\mathcal{A},\mathcal{B},\mathcal{C}\in\mathfrak{A}_{\mathbb{C}_\tau}$, $$[x\mathcal{A}+y\mathcal{B},\mathcal{C}]_\tau=x[\mathcal{A},\mathcal{C}]_\tau+y[\mathcal{B},\mathcal{C}]_\tau
\quad{\text{and}}\quad [\mathcal{A},x\mathcal{B}+y\mathcal{C}]_\tau=x[\mathcal{A},\mathcal{B}]_\tau+y[\mathcal{A},\mathcal{C}]_\tau.$$
	\item [(b)]\textit{Alternativity}: $[\mathcal{A},\mathcal{A}]_\tau=0$, for all  $\mathcal{A}\in\mathfrak{A}_{\mathbb{C}_\tau}$.
	\item [(c)] \textit{The Jacobi identity}: for all  $\mathcal{A},\mathcal{B},\mathcal{C}\in\mathfrak{A}_{\mathbb{C}_\tau}$. $$[\mathcal{A},[\mathcal{B},\mathcal{C}]_\tau]_\tau+[\mathcal{C},
[\mathcal{A},\mathcal{B}]_\tau]_\tau+[\mathcal{B},[\mathcal{C},\mathcal{A}]_\tau]_\tau=0.$$
	\item [(d)] \textit{Anti-commutativity}: $[\mathcal{A},\mathcal{B}]_\tau=-[\mathcal{B},\mathcal{A}]_\tau$, for all  $\mathcal{A},\mathcal{B}\in\mathfrak{A}_{\mathbb{C}_\tau}$.
\end{itemize}
Let $\mathcal{A}, \mathcal{B}\in\mathfrak{A}_{\mathbb{C}_\tau}$, then there exists $a,b,c,x,y,z\in\mathbb{C}_\tau$ such that
\begin{center}
	$\mathcal{A}=a\cdot\mathbb{I}_{\mathfrak{H}}+ b\cdot\mathsf{a}+  c\cdot \mathsf{a}^\dagger$	and $\mathcal{B}=x\cdot\mathbb{I}_{\mathfrak{H}}+ y\cdot\mathsf{a}+  z\cdot \mathsf{a}^\dagger$.
\end{center}
Then
\begin{center}

	\begin{tabular}{ c c c c c c c }
		$[\mathcal{A},\mathcal{B}]_\tau$	& $=$ & $ax\cdot[\mathbb{I}_{\mathfrak{H}},\mathbb{I}_{\mathfrak{H}}]_\tau$ & $+$ & $ay\cdot[\mathbb{I}_{\mathfrak{H}},\mathsf{a}]_\tau$ & $+$ & $az\cdot[\mathbb{I}_{\mathfrak{H}},\mathsf{a}^\dagger]_\tau$ \\\\
		~ & $+$ & $bx\cdot[\mathsf{a},\mathbb{I}_{\mathfrak{H}}]_\tau$ & $+$ & $by\cdot[\mathsf{a},\mathsf{a}]_\tau$ & $+$ & $bz\cdot[\mathsf{a},\mathsf{a}^\dagger]_\tau$ \\\\
		~ & $+$ & $cx\cdot[\mathsf{a}^\dagger,\mathbb{I}_{\mathfrak{H}}]_\tau$ & $+$ & $cy\cdot[\mathsf{a}^\dagger,\mathsf{a}]_\tau$ & $+$ & $cz\cdot[\mathsf{a}^\dagger,\mathsf{a}^\dagger]_\tau$.
	\end{tabular}
\end{center}
But $[\mathbb{I}_{\mathfrak{H}},\mathbb{I}_{\mathfrak{H}}]_\tau=[\mathsf{a},\mathsf{a}]_\tau=[\mathsf{a}^\dagger,\mathsf{a}^\dagger]_\tau=0$ and with the aid of Proposition \ref{xAq}, we can obtain that
$$[\mathbb{I}_{\mathfrak{H}},\mathsf{a}]_\tau=\mathbb{I}_{\mathfrak{H}}\mathsf{a}-\mathsf{a}\,\mathbb{I}_{\mathfrak{H}}=(\mathbb{I}_{\mathfrak{H}}\mathsf{a}-\mathsf{a}\mathbb{I}_{\mathfrak{H}})=[\mathbb{I}_{\mathfrak{H}},\mathsf{a}]_\tau=0.$$ Similarly $[\mathbb{I}_{\mathfrak{H}},\mathsf{a}^\dagger]_\tau=[\mathsf{a},\mathbb{I}_{\mathfrak{H}}]_\tau
=[\mathsf{a}^\dagger,\mathbb{I}_{\mathfrak{H}}]_\tau=0.$
Thus $[\mathcal{A},\mathcal{B}]_\tau=(bz-cy)\cdot\mathbb{I}_{\mathfrak{H}}\in\mathfrak{A}_{\mathbb{C}_\tau}$ and $[\cdot,\cdot]_\tau$ is a binary operation on $\mathfrak{A}_{\mathbb{C}_\tau}$. Hence $\mathfrak{A}_{\mathbb{C}_\tau}$ is a Lie algebra with the Lie bracket $[\cdot,\cdot]_\tau$. But it cannot be the complete version of a quaternionic Weyl-Heisenberg Lie algebra. It is a sub case of the quaternionic Weyl-Heisenberg Lie algebra, since $\mathfrak{A}_{\mathbb{C}_\tau}$ involves a single $\tau\in\{i,j,k\}$ at a time. Little more generally, for each $I\in\mathbb{S}$, we have
$$\mathfrak{A}_{\mathbb{C}_I}=\mbox{linear span over }\mathbb{C}_I\,\{I\cdot\mathbb{I}_{\mathfrak{H}}, Q, P_I\};$$
where $\mathbb{C}_I=\{x=x_1+I x_2~|~x_1, x_2\in\mathbb{R}\}$, is a Lie algebra associated with the Lie bracket $[\cdot, \cdot]_I$, where
$$[\mathcal{A}, \mathcal{B}]_I=\mathcal{A} \mathcal{B}-\mathcal{B} \mathcal{A},\mbox{~~for all~~}\mathcal{A}, \mathcal{B}\in\mathfrak{A}_{\mathbb{C}_I}.$$
\subsection{Quaternionic Weyl-Heisenberg Lie Algebra}
In this section we shall investigate the complete version of a quaternionic Weyl-Heisenberg Lie algebra. Only for notational convenience, in order to write a quaternion as $\qu=q_0+\sum_{\tau=i,j,k}q_{\tau}\tau$, we shall write $\qu=q_0+q_ii+q_jj+q_kk$ with $q_0,q_i,q_j,q_k\in\mathbb{R}$. At times, we shall also write $(q_0,q_i, q_j, q_k)\in\quat$. All these notations bear the same meaning $\qu=q_0+q_ii+q_jj+q_kk$ with $q_0,q_i,q_j,q_k\in\mathbb{R}$.\\
Let
\begin{equation}\label{X1}
\mathfrak{A}=\mbox{linear span over }\mathbb{R}\,\{\tau\cdot\mathbb{I}_{\mathfrak{H}}, Q, P_\tau~\mid~\tau=i,j,k\}.
\end{equation}
Then obviously $\mathfrak{A}$ is a vector space over $\mathbb{R}$. Define $[\cdot,\cdot]_\sigma:\mathfrak{A}\times\mathfrak{A}\longrightarrow\mathfrak{A}$ by
\begin{equation}\label{liebra}
[\mathcal{A},\mathcal{B}]_\sigma=\sum_{\tau=i,j,k}[\mathcal{A}_\tau,\mathcal{B}_\tau]_\tau;
\end{equation}
where for each $\tau=i,j,k$, $\mathcal{A}_\tau=x_\tau\tau\cdot\mathbb{I}_{\mathfrak{H}}+\dfrac{y}{\sqrt{3}} Q+z_\tau P_\tau$ and $\mathcal{B}_\tau=r_\tau\tau\cdot\mathbb{I}_{\mathfrak{H}}+\dfrac{s}{\sqrt{3}} Q+t_\tau P_\tau$ with $x_\tau,y,z_\tau,r_\tau,s,t_\tau\in\mathbb{R}$, and  $\displaystyle\mathcal{A}=x_\tau\tau\cdot\mathbb{I}_{\mathfrak{H}}+y\, Q+z_\tau P_\tau$,  $\displaystyle\mathcal{B}=r_\tau\tau\cdot\mathbb{I}_{\mathfrak{H}}+s\, Q+t_\tau P_\tau.$ Let $(\mathcal{A},\mathcal{B}),(\mathcal{C},\mathcal{D})\in\mathfrak{A}\times\mathfrak{A}$ with $(\mathcal{A},\mathcal{B})=(\mathcal{C},\mathcal{D})$. Then $\mathcal{A}=\mathcal{C}$ and $\mathcal{B}=\mathcal{D}$. But
\begin{equation}\label{eqA}
\mathcal{A}=\sum_{\tau=i,j,k}x_\tau\tau\cdot\mathbb{I}_{\mathfrak{H}}+yQ+\sum_{\tau=i,j,k}z_\tau P_{\tau},
\end{equation}

\begin{equation}\label{eqB}
\mathcal{B}=\sum_{\tau=i,j,k}r_\tau\tau\cdot\mathbb{I}_{\mathfrak{H}}+sQ+\sum_{\tau=i,j,k}t_\tau P_{\tau},
\end{equation}
\begin{equation}\label{eqC}
\mathcal{C}=\sum_{\tau=i,j,k}a_\tau\tau\cdot\mathbb{I}_{\mathfrak{H}}+bQ+\sum_{\tau=i,j,k}c_\tau P_{\tau}
\end{equation}
and
\begin{equation}\label{eqD}
\mathcal{D}=\sum_{\tau=i,j,k}l_\tau\tau\cdot\mathbb{I}_{\mathfrak{H}}+mQ+\sum_{\tau=i,j,k}n_\tau P_{\tau},
\end{equation}
for some $x_\tau,y,z_\tau,r_\tau,s,t_\tau,a_\tau,b,c_\tau,l_\tau,m,n_\tau\in\mathbb{R}:\tau=i,j,k$. Thus, for each $\tau=i,j,k$, $x_\tau=a_\tau,y=b,z_\tau=c_\tau,r_\tau=l_\tau,s=m$ and $t_\tau=n_\tau$. This implies that for each $\tau=i,j,k$, $\mathcal{A}_\tau=\mathcal{C}_\tau$ and $\mathcal{B}_\tau=\mathcal{D}_\tau$  for any $\tau=i,j,k$, $\mathcal{A}_\tau=x_\tau\tau\cdot\mathbb{I}_{\mathfrak{H}}+\dfrac{y}{\sqrt{3}} Q+z_\tau P_\tau$, $\mathcal{B}_\tau=r_\tau\tau\cdot\mathbb{I}_{\mathfrak{H}}+\dfrac{s}{\sqrt{3}} Q+t_\tau P_\tau$,
$\mathcal{C}_\tau=a_\tau\tau\cdot\mathbb{I}_{\mathfrak{H}}+\dfrac{b}{\sqrt{3}} Q+c_\tau P_\tau$ and $\mathcal{D}_\tau=l_\tau\tau\cdot\mathbb{I}_{\mathfrak{H}}+\dfrac{m}{\sqrt{3}} Q+n_\tau P_\tau$. Since for each $\tau=i,j,k$, $[\cdot,\cdot]_\tau$ is well-defined, we have $[\mathcal{A}_\tau,\mathcal{B}_\tau]_\tau=[\mathcal{C}_\tau,\mathcal{D}_\tau]_\tau$, for all $\tau=i,j,k$. Hence $[\mathcal{A},\mathcal{B}]_\sigma=[\mathcal{C},\mathcal{D}]_\sigma$.
Further, $$[\mathcal{A},\mathcal{B}]_\sigma=\sum_{\tau=i,j,k}[\mathcal{A}_\tau,\mathcal{B}_\tau]_\tau=\sum_{\tau=i,j,k}\dfrac{1}{\sqrt{3}}(yt_\tau-z_\tau s)\tau\cdot\mathbb{I}_{\mathfrak{H}}\in\mathfrak{A}.$$
Therefore $[\cdot,\cdot]_\sigma$ is well-defined. For each $\tau=i,j,k$, $[\cdot,\cdot]_\tau$ satisfies the axioms (a)-(d). Therefore, for any $\displaystyle\mathcal{A}$, $\displaystyle\mathcal{B}$, $\displaystyle\mathcal{C}\in\mathfrak{A}$ as in (\ref{eqA})-(\ref{eqC}) and $x,y\in\mathbb{R}$ we have $$[x\mathcal{A}+y\mathcal{B},\mathcal{C}]_\sigma=
\sum_{\tau=i,j,k}[x\mathcal{A}_\tau+y\mathcal{B}_\tau,\mathcal{C}_\tau]_\tau
=\sum_{\tau=i,j,k}x[\mathcal{A}_\tau,\mathcal{C}_\tau]_\tau+\sum_{\tau
=i,j,k}y[\mathcal{B}_\tau,\mathcal{C}_\tau]_\tau
=x[\mathcal{A},\mathcal{C}]_\sigma+y[\mathcal{B},\mathcal{C}]_\sigma$$ and similarly we can obtain  $$[\mathcal{A},x\mathcal{B}+y\mathcal{C}]_\sigma=x[\mathcal{A},\mathcal{B}]_\sigma+y[\mathcal{A},\mathcal{C}]_\sigma.$$ This shows the bilinearity of $[\cdot,\cdot]_{\sigma}$.
Let $\displaystyle\mathcal{A}\in\mathfrak{A}$ as in (\ref{eqA}). Since $[\mathcal{A}_\tau,\mathcal{A}_\tau]_\tau=0$, for all $\tau=i,j,k$ we have $[\mathcal{A},\mathcal{A}]_\sigma=\displaystyle\sum_{\tau=i,j,k}[\mathcal{A}_\tau,\mathcal{A}_\tau]_\tau=0$ and hence the alternativity of $[\cdot,\cdot]_{\sigma}$ follows. Let $\displaystyle\mathcal{A}$, $\displaystyle\mathcal{B}$, $\displaystyle\mathcal{C}\in\mathfrak{A}$ as in (\ref{eqA})-(\ref{eqC}).
Now using the Jacobian identities of $[\cdot,\cdot]_\tau:\tau=i,j,k$, we have
$$\sum_{\tau=i,j,k}\left[ \mathcal{A}_\tau,\left[ \mathcal{B}_\tau,\mathcal{C}_\tau\right]_\tau \right]_\tau=-\sum_{\tau=i,j,k} ([\mathcal{C}_\tau,[\mathcal{A}_\tau,\mathcal{B}_\tau]_\tau]_\tau+
[\mathcal{B}_\tau,[\mathcal{C}_\tau,\mathcal{A}_\tau]_\tau]_\tau)$$ which implies
$$[\mathcal{A},[\mathcal{B},\mathcal{C}]_\sigma]_\sigma
=-[\mathcal{C},[\mathcal{A},\mathcal{B}]_\sigma]_\sigma-[\mathcal{B},[\mathcal{C},\mathcal{A}]_\sigma]_\sigma.$$ This concludes the Jacobian identity of $[\cdot,\cdot]_{\sigma}$. Let $\displaystyle\mathcal{A}$, $\displaystyle\mathcal{B}\in\mathfrak{A}$ as in (\ref{eqA}) and (\ref{eqB}). The anti-commutativity of $[\cdot,\cdot]_\tau:\tau=i,j,k$ gives us the anti-commutativity of $[\cdot,\cdot]_{\sigma}$, that is,
$$[\mathcal{A},\mathcal{B}]_\sigma=\sum_{\tau=i,j,k}[\mathcal{A}_\tau,\mathcal{B}_\tau]_\tau
=-\sum_{\tau=i,j,k}[\mathcal{B}_\tau,\mathcal{A}_\tau]_\tau=-[\mathcal{B},\mathcal{A}]_\sigma.$$
 Therefore $\mathfrak{A}$ is a Lie algebra with the Lie bracket $[\cdot,\cdot]_{\sigma}$. We call this Lie algebra a Weyl-Heisenberg Lie Algebra.\\

But unfortunately $\mathsf{a},\mathsf{a}^\dagger\notin\mathfrak{A}$, this prevents us to use the commutator  $[\mathsf{a},\mathsf{a}^\dagger]_{\sigma}$ legitimately. In order to avoid this difficulty, we have to find a bigger Lie-algebra which can hold $\mathsf{a},\mathsf{a}^\dagger$ as elements and the Lie-algebra $\mathfrak{A}$ should become a sub Lie-algebra of it. Let $$\mathfrak{A}_{\quat}=\mbox{linear span over }\,\quat\,\{\mathbb{I}_{\mathfrak{H}}, \mathsf{a}, \mathsf{a}^\dagger\}.$$
Then the Proposition \ref{lft_mul} guarantees, together with the Remark \ref{Rem123}, that  $\mathfrak{A}_{\quat}$ is a vector space over $\quat$, and it contains $\mathfrak{A}$, under the left multiplication \textquoteleft$\cdot$\textquoteright \,which is defined in (\ref{lft_mul-op}). The extended map of the map (\ref{liebra}), $[\cdot,\cdot]:\mathfrak{A}_{\quat}\times\mathfrak{A}_{\quat}\longrightarrow\mathfrak{A}_{\quat},$ can be defined by
\begin{equation}\label{liebraH}
[\mathcal{A},\mathcal{B}]=\sum_{\tau=i,j,k}[\mathcal{A}_\tau,\mathcal{B}_\tau]_\tau,
\end{equation}
where $\mathcal{A}=\mathfrak{x}\cdot\mathbb{I}_{\mathfrak{H}}+\mathfrak{y}\cdot\mathsf{a}
+\mathfrak{z}\cdot\mathsf{a}^\dagger$, $\mathcal{B}=\mathfrak{r}\cdot\mathbb{I}_{\mathfrak{H}}+\mathfrak{s}\cdot\mathsf{a}
+\mathfrak{t}\cdot\mathsf{a}^\dagger$ with $\mathfrak{x}=(x_0,x_i,x_j,x_k),\mathfrak{y}=(y_0,y_i,y_j,y_k),\mathfrak{z}
=(z_0,z_i,z_j,z_k),\mathfrak{r}=(r_0,r_i,r_j,r_k),
\mathfrak{s}=(s_0,s_i,s_j,s_k),\mathfrak{t}=(t_0,t_i,t_j,t_k)\in\quat$. For each $\tau=i,j,k$, we can write
\begin{eqnarray*}
\mathcal{A}_\tau &=&(\dfrac{x_0}{\sqrt{3}}+x_\tau\,\tau)\cdot\mathbb{I}_{\mathfrak{H}}
+(\dfrac{y_0}{\sqrt{3}}+y_\tau\,\tau)\cdot \mathsf{a}+(\dfrac{z_0}{\sqrt{3}}+z_\tau\,\tau)\cdot \mathsf{a}^\dagger\quad{\text{ and}} \\ \mathcal{B}_\tau &=&(\dfrac{r_0}{\sqrt{3}}+r_\tau\,\tau)\cdot\mathbb{I}_{\mathfrak{H}}
+(\dfrac{s_0}{\sqrt{3}}+s_\tau\,\tau)\cdot \mathsf{a}+(\dfrac{t_0}{\sqrt{3}}+t_\tau\,\tau)\cdot \mathsf{a}^\dagger.
\end{eqnarray*}
Then, for any $\tau=i,j,k$, we have
\begin{eqnarray*}
	[\mathcal{A}_\tau,\mathcal{B}_\tau]_\tau
	&=&[(\dfrac{y_0}{\sqrt{3}}+y_\tau\,\tau)\cdot \mathsf{a},(\dfrac{t_0}{\sqrt{3}}+t_\tau\,\tau)\cdot \mathsf{a}^\dagger]_\tau+[(\dfrac{z_0}{\sqrt{3}}+z_\tau\,\tau)\cdot \mathsf{a}^\dagger,(\dfrac{s_0}{\sqrt{3}}+s_\tau\,\tau)\cdot \mathsf{a}]_\tau\\ &=&\left[(\dfrac{y_0}{\sqrt{3}}+y_\tau\,\tau)(\dfrac{t_0}{\sqrt{3}}+t_\tau\,\tau)-(\dfrac{z_0}{\sqrt{3}}+z_\tau\,\tau)(\dfrac{s_0}{\sqrt{3}}+s_\tau\,\tau) \right]\cdot\mathbb{I}_{\mathfrak{H}}\mbox{~~as~~}[\mathsf{a},\mathsf{a}^\dagger]=\mathbb{I}_{\mathfrak{H}}.
\end{eqnarray*}
This shows $$[\mathcal{A},\mathcal{B}]=\left( \sum_{\tau=i,j,k}\left[(\dfrac{y_0}{\sqrt{3}}+y_\tau\,\tau)(\dfrac{t_0}{\sqrt{3}}+t_\tau\,\tau)-(\dfrac{z_0}{\sqrt{3}}+z_\tau\,\tau)(\dfrac{s_0}{\sqrt{3}}+s_\tau\,\tau) \right]\right)\cdot\mathbb{I}_{\mathfrak{H}}\in\mathfrak{A}_{\quat}.$$
Further, since $\mathfrak{A}_\tau$ is a Lie-algebra with the Lie-bracket $[\cdot,\cdot]_\tau$, we can see that the bracket $[\cdot,\cdot]$ satisfies the axioms: bilinearity, alternativity, the Jacobi identity, and the anti-commutativity. Hence $\mathfrak{A}_{\quat}$ is a Lie-algebra with the Lie-bracket $[\cdot,\cdot]$. Let $\displaystyle\mathcal{A}$, $\displaystyle\mathcal{B}\in\mathfrak{A}$ as in (\ref{eqA}) and (\ref{eqB}). Then
$$\mathcal{A}=(x_i+x_j+x_k)\cdot\mathbb{I}_{\mathfrak{H}}+\dfrac{1}{\sqrt{2}}(y-z_i-z_j-z_k)\cdot \mathsf{a}+\dfrac{1}{\sqrt{2}}(y+z_i+z_j+z_k)\cdot \mathsf{a}^\dagger$$
and $$
\mathcal{B}=(r_i+r_j+r_k)\cdot\mathbb{I}_{\mathfrak{H}}+\dfrac{1}{\sqrt{2}}(s-t_i-t_j-t_k)\cdot \mathsf{a}+\dfrac{1}{\sqrt{2}}(s+t_i+t_j+ t_k)\cdot \mathsf{a}^\dagger.$$
Thus \begin{eqnarray*}
	[\mathcal{A},\mathcal{B}]
	&=& \sum_{\tau=i,j,k}\left[\dfrac{1}{\sqrt{2}}(\dfrac{y}{\sqrt{3}}-z_\tau\,\tau)\dfrac{1}{\sqrt{2}}(\dfrac{s}{\sqrt{3}}+t_\tau\,\tau)-\dfrac{1}{\sqrt{2}}(\dfrac{y}{\sqrt{3}}+z_\tau\,\tau)\dfrac{1}{\sqrt{2}}(\dfrac{s}{\sqrt{3}}-t_\tau\,\tau) \right]\\
	&=&\sum_{\tau=i,j,k}\dfrac{1}{\sqrt{3}}(yt_\tau-z_\tau s)\cdot\mathbb{I}_{\mathfrak{H}}=[\mathcal{A},\mathcal{B}]_\sigma\in\mathfrak{A}_{\quat}.
\end{eqnarray*}
Hence, two things became clear: one is that $\mathfrak{A}$ is a sub Lie-algebra of $\mathfrak{A}_{\quat}$, and other one is that $[\cdot,\cdot]$ is an extended map of $[\cdot,\cdot]_\sigma$.

For $\tau\in\{i,j,k\}$, take
 \begin{equation}\label{X1tau}
 \mathfrak{A}_\tau=\mbox{linear span over }\mathbb{R}\,\{\tau\cdot\mathbb{I}_{\mathfrak{H}}, Q, P_\tau\}.
 \end{equation}
Then $\mathfrak{A}_\tau$ is a vector space over $\mathbb{R}$ and a subspace of $\mathfrak{A}_{\mathbb{C}_\tau}$. $\mathfrak{A}_\tau$ is also a sub Lie algebra of $\mathfrak{A}_{\mathbb{C}_\tau}$ with the Lie bracket $[\cdot,\cdot]_\tau$.

The following Proposition shows us that the Lie algebra $\mathfrak{A}$ defined in (\ref{X1}) is embedded in the direct sum $\displaystyle\bigoplus_{\tau=i,j,k}\mathfrak{A}_\tau$. That is, the Lie algebra $\mathfrak{A}$ has a similar Lie structure of $\displaystyle\bigoplus_{\tau=i,j,k}\mathfrak{A}_\tau$. Indeed, this gives more meaning and validation to the definition of the Lie bracket $[\cdot,\cdot]_{\sigma}$ given in Eq. (\ref{liebra}).
\begin{proposition}
Let $\mathfrak{A}$, $\mathfrak{A}_\tau$ as in (\ref{X1}), (\ref{X1tau}) respectively.
	The map $\displaystyle\sigma:\mathfrak{A}\longrightarrow\bigoplus_{\tau=i,j,k}\mathfrak{A}_\tau$, defined by
	$$\sigma(\mathcal{A})=(\mathcal{A}_\tau),~~\mbox{for all}~~\mathcal{A}=\sum_{\tau=i,j,k}\mathcal{A}_\tau\in\mathfrak{A},$$
	where  $\displaystyle\mathcal{A}=\sum_{\tau=i,j,k}x_\tau\tau\cdot\mathbb{I}_{\mathfrak{H}}+yQ+\sum_{\tau=i,j,k}z_\tau P_{\tau}$ with $x_\tau,y,z_\tau\in\mathbb{R}~~\mbox{and for each}~~\tau=i,j,k,~~
 \mathcal{A}_\tau=x_\tau\tau\cdot\mathbb{I}_{\mathfrak{H}}+\dfrac{y}{\sqrt{3}} Q+z_\tau P_\tau$, is an embedding.
\end{proposition}
\begin{proof}
	Let $\displaystyle\mathcal{A},\mathcal{B}\in\mathfrak{A}$ as in (\ref{eqA}),(\ref{eqB}), and for each $\tau=i,j,k,$, $\mathcal{A}_\tau=x_\tau\tau\cdot\mathbb{I}_{\mathfrak{H}}+\dfrac{y}{\sqrt{3}} Q+z_\tau P_\tau$, $\mathcal{B}_\tau=r_\tau\tau\cdot\mathbb{I}_{\mathfrak{H}}+\dfrac{s}{\sqrt{3}} Q+t_\tau P_\tau$ with $x_\tau,y,z_\tau,r_\tau,s,t_\tau\in\mathbb{R}$. If $\mathcal{A}=\mathcal{B}$, then it is easy to see that $\mathcal{A}_\tau=\mathcal{B}_\tau$, for all $\tau=i,j,k,$. This implies that $\sigma(\mathcal{A})=(\mathcal{A}_\tau)=(\mathcal{B}_\tau)=\sigma(\mathcal{B}).$
On the other hand, we can also obtain that, if $\sigma(\mathcal{A})=(\mathcal{A}_\tau)=(\mathcal{B}_\tau)=\sigma(\mathcal{B})$, then $\mathcal{A}=\mathcal{B}$. Further, one can trivially say that $\displaystyle\sigma(\mathfrak{A})\subseteq\bigoplus_{\tau=i,j,k}\mathfrak{A}_\tau$. Thus $\sigma$ is well-defined and injective. Therefore the inverse of $\sigma$, $\displaystyle\sigma^{-1}:\sigma(\mathfrak{A})\longrightarrow\mathfrak{A}$ exists and it is injective. Now we shall show that $\sigma,~\sigma^{-1}$ are Lie homomorphism. For, let $\displaystyle\mathcal{A},\mathcal{B}\in\mathfrak{A}$  as in (\ref{eqA}),(\ref{eqB}) and for each $\tau=i,j,k,$  $\mathcal{A}_\tau=x_\tau\tau\cdot\mathbb{I}_{\mathfrak{H}}+\dfrac{y}{\sqrt{3}} Q+z_\tau P_\tau$, $\mathcal{B}_\tau=r_\tau\tau\cdot\mathbb{I}_{\mathfrak{H}}+\dfrac{s}{\sqrt{3}} Q+t_\tau P_\tau$ with $x_\tau,y,z_\tau,r_\tau,s,t_\tau\in\mathbb{R}$, then $$[\mathcal{A},\mathcal{B}]=\sum_{\tau=i,j,k}[\mathcal{A}_\tau,\mathcal{B}_\tau]_\tau=\sum_{\tau=i,j,k}\dfrac{1}{\sqrt{3}}(yt_\tau-z_\tau s)\tau\cdot\mathbb{I}_{\mathfrak{H}}$$
	and this implies that $\sigma([\mathcal{A},\mathcal{B}])=([\mathcal{A}_\tau,\mathcal{B}_\tau]_\tau)$. But $\sigma(\mathcal{A})=(\mathcal{A}_\tau)$ and $\sigma(\mathcal{B})=(\mathcal{B}_\tau)$. So $[\sigma(\mathcal{A}),\sigma(\mathcal{B})]_\oplus=([\mathcal{A}_\tau,\mathcal{B}_\tau]_\tau)$, where $[\cdot,\cdot]_\oplus$, is the typical Lie bracket of
	$\displaystyle\bigoplus_{\tau=i,j,k}\mathfrak{A}_\tau$, and it is defined by $[\mathcal{X},\mathcal{Y}]_\oplus=([\mathcal{X}_\tau,\mathcal{Y}_\tau]_\tau)$, for all $\mathcal{X}=(\mathcal{X}_\tau), \mathcal{Y}=(\mathcal{Y}_\tau)\in\displaystyle\bigoplus_{\tau=i,j,k}\mathfrak{A}_\tau$. Therefore $$\sigma([\mathcal{A},\mathcal{B}])=[\sigma(\mathcal{A}),\sigma(\mathcal{B})]_\oplus$$ and $\sigma$ is a Lie homomorphism. Now let $\displaystyle\mathcal{A}=(\mathcal{A}_\tau),\mathcal{B}=(\mathcal{B}_\tau)\in\sigma(\mathfrak{A})$ and for each $\tau=i,j,k,$, $\mathcal{A}_\tau=x_\tau\tau\cdot\mathbb{I}_{\mathfrak{H}}+\dfrac{y}{\sqrt{3}} Q+z_\tau P_\tau$, $\mathcal{B}_\tau=r_\tau\tau\cdot\mathbb{I}_{\mathfrak{H}}+\dfrac{s}{\sqrt{3}} Q+t_\tau P_\tau$ with $x_\tau,y,z_\tau,r_\tau,s,t_\tau\in\mathbb{R}$, then $\displaystyle\sigma^{-1}([\mathcal{A},\mathcal{B}]_\oplus)=\sum_{\tau=i,j,k}[\mathcal{A}_\tau,\mathcal{B}_\tau]_\tau$. But $$\displaystyle\sigma^{-1}(\mathcal{A})=\sum_{\tau=i,j,k}x_\tau\tau\cdot\mathbb{I}_{\mathfrak{H}}+yQ+\sum_{\tau=i,j,k}z_\tau P_{\tau},\sigma^{-1}(\mathcal{B})=\sum_{\tau=i,j,k}r_\tau\tau\cdot\mathbb{I}_{\mathfrak{H}}+sQ+\sum_{\tau=i,j,k}t_\tau P_{\tau}\in\mathfrak{A}$$
 and therefore
$$[\sigma^{-1}(\mathcal{A}),\sigma^{-1}(\mathcal{B})]=
\displaystyle\sum_{\tau=i,j,k}[\mathcal{A}_\tau,\mathcal{B}_\tau]_\tau.$$
Thus
	$$\sigma^{-1}([\mathcal{A},\mathcal{B}]_\oplus)=[\sigma^{-1}(\mathcal{A}),\sigma^{-1}(\mathcal{B})].$$ Hence $\sigma$ is an embedding.
\end{proof}
\subsection{The Displacement Operator}
As we have explained in the introduction, on a right quaternionic Hilbert space with a right multiplication we cannot have a displacement operator as a representation for the representation space $\HI$. This fact has been indicated twice in the literature, in \cite{Ad2} while studying quaternionic Perelomov type CS and in \cite{Thi2} when the authors studied the quaternionic canonical CS. However, in this subsection, we shall show that if we consider a right quaternionic Hilbert space with a left multiplication on it, see Eq. (\ref{lft_mul-op}), we can have a displacement operator as a representation for the representation space $\HI$ with all the desired properties.\\

 We have already obtained, from
$$\mid\gamma_{\qu}\rangle=e^{-|\qu|^2/2}\cdot\left[ \sum_{n=0}^\infty|e_n\rangle\dfrac{\qu^n}{\sqrt{n!}}\right]$$
and
$$\mid e_n\rangle=\frac{({\mathsf{a}}^\dagger )^n}{\sqrt{n!}}\mid e_0\rangle,$$
that
$$\mid\gamma_{\qu}\rangle=e^{-|\qu|^2/2}\cdot\left[ \sum_{n=0}^\infty\frac{(\mathsf{a}^\dagger )^n}{\sqrt{n!}}\mid e_0\rangle\right]=(e^{-|\qu|^2/2}\cdot e^{\qu\cdot {\mathsf{a}}^\dagger })|e_0\rangle.$$
We consider an alternative version of the Lie algebra $\mathfrak{A}$ given in (\ref{X1}), denoted by $\hat{\mathfrak{A}}$,  as the quaternionic Weyl-Heisenberg Lie algebra
\begin{equation}\label{alt_WHL_alg}
\hat{\mathfrak{A}}=\mbox{linear span over }\mathbb{R}\,\{\tau\cdot\mathbb{I}_{\mathfrak{H}}, \tau\cdot Q, P_0~\mid~\tau=i,j,k\},
\end{equation}
where $P_0=-\left( \dfrac{\mathsf{a}-\mathsf{a}^\dagger}{\sqrt{2}}\right) =\tau\cdot P_\tau,$ for all $\tau=i,j,k$. It is a sub Lie algebra of $\mathfrak{A}_\quat$.
Here the Lie bracket, $[\![\cdot,\cdot]\!]:\hat{\mathfrak{A}}\times\hat{\mathfrak{A}}\longrightarrow\hat{\mathfrak{A}}$, of $\hat{\mathfrak{A}}$, is defined   by
$$[\![\mathcal{A},\mathcal{B}]\!]=\sum_{\tau=i,j,k}[\mathcal{A}_\tau,\mathcal{B}_\tau]_\tau=-\sum_{\tau=i,j,k}\dfrac{1}{\sqrt{3}}(yt_\tau-z_\tau s)\tau\cdot\mathbb{I}_{\mathfrak{H}}\in\hat{\mathfrak{A}};$$
for all $\displaystyle\mathcal{A}=\sum_{\tau=i,j,k}x_\tau\tau\cdot\mathbb{I}_{\mathfrak{H}}+yP_0+\sum_{\tau=i,j,k}z_\tau \tau\cdot Q$, $\displaystyle\mathcal{B}=\sum_{\tau=i,j,k}r_\tau\tau\cdot\mathbb{I}_{\mathfrak{H}}+sP_0+\sum_{\tau=i,j,k}t_\tau \tau\cdot Q\in\hat{\mathfrak{A}}$, where for each $\tau=i,j,k$, $\mathcal{A}_\tau=x_\tau\tau\cdot\mathbb{I}_{\mathfrak{H}}+\dfrac{y}{\sqrt{3}} P_0+z_\tau\,\tau\cdot Q$ and $\mathcal{B}_\tau=r_\tau\tau\cdot\mathbb{I}_{\mathfrak{H}}+\dfrac{s}{\sqrt{3}} P_0+t_\tau \,\tau\cdot Q$ with $x_\tau,y,z_\tau,r_\tau,s,t_\tau\in\mathbb{R}$.
\begin{remark}
The Lie bracket $[\![\cdot,\cdot]\!]$ is a restricted version of the Lie bracket $[\cdot,\cdot]$. In fact, both Lie brackets $[\cdot,\cdot]_\sigma$ and $[\![\cdot,\cdot]\!]$ are restrictions of $[\cdot,\cdot]$. This fact enables us to use the same notation $[\cdot,\cdot]$ for all three Lie algebras, $\mathfrak{A}_\sigma$, $\hat{\mathfrak{A}}$ and $\mathfrak{A}_\quat$.
\end{remark}
A general element $\mathcal{A}$ of $\hat{\mathfrak{A}}$ can be written as
\begin{equation}\label{gen_ele1}
\mathcal{A}=xi\cdot\mathbb{I}_{\mathfrak{H}}+yj\cdot\mathbb{I}_{\mathfrak{H}}+zk\cdot\mathbb{I}_{\mathfrak{H}}+3aP_0+bi\cdot Q+cj\cdot Q+dk\cdot Q={\bf{x}}\cdot\mathbb{I}_{\mathfrak{H}}+(\mathfrak{q}\cdot \mathsf{a}^\dagger-\oqu\cdot \mathsf{a})
\end{equation} and
\begin{equation}\label{gen_ele2}
\mathcal{A}= (xi\cdot\mathbb{I}_{\mathfrak{H}}+i\cdot(bQ-aP_i))+(yj\cdot\mathbb{I}_{\mathfrak{H}}+j\cdot(cQ-aP_j))+(zk\cdot\mathbb{I}_{\mathfrak{H}}+k\cdot(dQ-aP_k));
\end{equation}
where ${\bf{x}}=xi+yj+zk$ and $\mathfrak{q}=\dfrac{1}{\sqrt{2}}(3a+bi+cj+dk)$ in the equation (\ref{gen_ele1}) with $x,y,z,a,b,c,d\in\mathbb{R}$.
The $\mathcal{A}$ in (\ref{gen_ele1}) can also be written as
$$\mathcal{A}=xi\cdot\mathbb{I}_{\mathfrak{H}}+(\mathfrak{q}_i\cdot \mathsf{a}^\dagger-{\oqu_i}\cdot \mathsf{a})+yj\cdot\mathbb{I}_{\mathfrak{H}}+(\mathfrak{q}_j\cdot \mathsf{a}^\dagger-{\oqu_j}\cdot \mathsf{a})+zk\cdot\mathbb{I}_{\mathfrak{H}}+(\mathfrak{q}_k\cdot \mathsf{a}^\dagger-{\oqu_k}\cdot \mathsf{a});$$
where $\mathfrak{q}_i=\dfrac{1}{\sqrt{2}}(a+ib)$, $\mathfrak{q}_j=\dfrac{1}{\sqrt{2}}(a+jc)$, $\mathfrak{q}_k=\dfrac{1}{\sqrt{2}}(a+kd)$ and the $\qu$ in (\ref{gen_ele1}) is $\mathfrak{q}=\mathfrak{q}_i+\mathfrak{q}_j+\mathfrak{q}_k$. Applying Equation (\ref{sc_mul_aj-op}) with the Proposition \ref{xAq} we can see that the operator $\mathcal{A}$ is anti-self-adjoint in $\mathfrak{H}$, and it is the infinitesimal generator of the operator:
\begin{equation}\label{infi}
U(\mathcal{A} )=e^{\mathcal{A}}=e^{{\bf{x}}}\,e^{(\mathfrak{q}\cdot \mathsf{a}^\dagger-\oqu\cdot \mathsf{a})}:=e^{\bf{x}}\,\D(\mathfrak{q});
\end{equation}
where $\D(\mathfrak{q})$ is the operator-valued map $\mathfrak{q}\longmapsto\D(\mathfrak{q}):=e^{(\mathfrak{q}\cdot \mathsf{a}^\dagger-\oqu\cdot \mathsf{a})}$.
\begin{proposition}\label{uni_op}
	For any $\mathcal{A}\in\hat{\mathfrak{A}}$, the operator $U(\mathcal{A})=e^{\mathcal{A}}$ is unitary and continuous. 	
\end{proposition}
\begin{proof} Let $\mathcal{A}\in\hat{\mathfrak{A}}$.
	The Baker-Campbell-Hausdorff identity (see, e.g., \cite{Thi2}) is
	\begin{equation}\label{B-C-H-I}
	e^{\mathcal{A}+\mathcal{B}}=e^{-\frac{1}{2}[\mathcal{A},\mathcal{B}]}\,e^{\mathcal{A}}e^{\mathcal{B}}
	\end{equation} when $\mathcal{A}$ and $\mathcal{B}$ commute with $[\mathcal{A},\mathcal{B}]$, i.e. $[\mathcal{A},[\!\mathcal{A},\mathcal{B}] ]=0=[\mathcal{B},[\mathcal{A},\mathcal{B}] ]$. Indeed, for any $\tau\in\{i,j,k\}$, $[\mathcal{A}_\tau,[\mathcal{A}_\tau,\mathcal{B}_\tau]_\tau]_\tau=0$. Thus $$[\mathcal{A},[\!\mathcal{A},\mathcal{B}] ]=\sum_{\tau=i,j,k}[A_\tau,[\mathcal{A}_\tau,\mathcal{B}_\tau]_\tau]_\tau=0.$$ Similarly $[\mathcal{B},[\mathcal{A},\mathcal{B}] ]=0$. Since we are working on the right quaternionic Hilbert space $\mathfrak{H}$ and the factor $e^{-\frac{1}{2}[\mathcal{A},\mathcal{B}]}$ is a quaternion scalar, one should note that the equation (\ref{B-C-H-I}) will be determined as follows:
	$$e^{\mathcal{A}+\mathcal{B}}\phi=(e^{\mathcal{A}}e^{\mathcal{B}}\phi)\,e^{-\frac{1}{2}[\mathcal{A},\mathcal{B}]},$$
	for all $\phi\in\mathfrak{H}$. Using this identity, we can obtain that
	$$U(\mathcal{A}){U(\mathcal{A})}^{\dagger}=e^{\mathcal{A}}{e^{\mathcal{A}}}^{\dagger}=e^{\mathcal{A}}e^{\mathcal{-A}}=e^{\frac{1}{2}[\mathcal{A},-\mathcal{A}]}\,e^{\mathcal{A}+(-\mathcal{A})}=\mathbb{I}_{\mathfrak{H}}.$$ Thus $U(\mathcal{A})$ is unitary.
	From this, for any $\mathcal{A}\in\hat{\mathfrak{A}}$ and $\phi,\psi\in\mathfrak{H}$, we have
	$$\|U(\mathcal{A})\phi-U(\mathcal{A})\psi\|=\|\phi-\psi\|.$$ It is enough to imply the continuity of $U(\mathcal{A})$. 	
\end{proof}
\begin{proposition}\label{uni_re}
	For any $\mathfrak{q},\mathfrak{p}\in\quat$, the composition rule of $\D(\mathfrak{q})$ and $\D(\mathfrak{p})$ is given by
	\begin{equation}\label{com_rl_D}
	\D(\mathfrak{q})\D(\mathfrak{p})=e^{-\sum_{\tau=i,j,k}\tau(\mathfrak{q}_{\tau}\wedge\mathfrak{p}_{\tau})}\,\D(\mathfrak{q}+\mathfrak{p}),
	\end{equation}
	where $\displaystyle \D(\qu)=e^{\mathfrak{q}\cdot \mathsf{a}^\dagger-\oqu\cdot \mathsf{a}}$ and $\displaystyle\D(\pu)=e^{\mathfrak{p}\cdot \mathsf{a}^\dagger-\bar{\mathfrak{p}}\cdot \mathsf{a}}$. Moreover, $\D(\mathfrak{q})$ is a unitary representation up to the phase factor $e^{-\sum_{\tau=i,j,k}\tau(\mathfrak{q}_{\tau}\wedge\mathfrak{p}_{\tau})}$ of the representation space $\mathfrak{H}$.	
\end{proposition}
\begin{proof}
	Let $\mathfrak{q},\mathfrak{p}\in\quat$ with $\mathcal{A}=(\mathfrak{q}\cdot \mathsf{a}^\dagger-\oqu\cdot \mathsf{a})$ and $\mathcal{B}=(\mathfrak{p}\cdot \mathsf{a}^\dagger-\bar{\mathfrak{p}}\cdot \mathsf{a})$. Now, using the Baker-Campbell-Hausdorff identity, we have that
	\begin{eqnarray*}
		\D(\mathfrak{q})\D(\mathfrak{p})
		&=& e^{\mathcal{A}}e^{\mathcal{B}} =e^{\frac{1}{2}[\mathcal{A},\mathcal{B}]}\,e^{\mathcal{A}+\mathcal{B}}\\
		&=& e^{-\sum_{\tau=i,j,k}\tau(\mathfrak{q}_{\tau}\wedge\mathfrak{p}_{\tau})}\,e^{((\mathfrak{q}+\mathfrak{p})\cdot \mathsf{a}^\dagger-\overline{(\mathfrak{q}+\mathfrak{p})}\cdot \mathsf{a})}\\
		&=& e^{-\sum_{\tau=i,j,k}\tau(\mathfrak{q}_{\tau}\wedge\mathfrak{p}_{\tau})}\,\D(\mathfrak{q}+\mathfrak{p}).
	\end{eqnarray*}
	Since $\D(\mathfrak{q})$ is a unitary operator, we can conclude that $\D(\mathfrak{q})$ is a unitary representation up to the phase factor $e^{-\sum_{\tau=i,j,k}\tau(\mathfrak{q}_{\tau}\wedge\mathfrak{p}_{\tau})}$ of the representation space $\mathfrak{H}$, the desired conclusion.
\end{proof}
The above proposition also says, apart from a phase factor, that the product of two displacement operators produces another displacement operator and it has the total displacement as the sum of two individual displacements.

In complex quantum mechanics, the vector $e^{i\theta}\phi$, $\theta\in\mathbb{R}$, represents the same physical state as $\phi$. Thus it is natural to consider projective unitary representations. In the Lie algebra set up a projective representation is a Lie algebra homomorphism of the Lie algebra onto the space $AU(\HI)/\{e^{i\theta}I\}$, $\theta\in\mathbb{R}$, where the Lie algebra $AU(\HI)$ is the space of anti-self-adjoint operators on $\HI$, the space $\{e^{i\theta}I\}$ is an ideal in $AU(\HI)$ and the quotient is in the sense of Lie algebra over $\mathbb{R}$ \cite{Brian}. Further, in quantum theory the overall phase of a quantum state is not an observable. Thus it is natural to consider projective representations because quantum theory predicts that states in a Hilbert space do not need to transform under rotations, but only under projective representations. For the quaternionic quantum mechanics a reasonable argument along these lines is given in \cite{Ad, Ad2}.\\

In our case, the above proposition establishes that the operator $\D$ is a homomorphism from the algebra $\quat$ of quaternions under the operation of usual addition. Furthermore, it should be noticed that, from the Baker-Campbell-Hausdorff identity, for any $\mathcal{A},\mathcal{B}\in\hat{\mathfrak{A}}$ and $\phi\in\mathfrak{H}$,
		\begin{equation}\label{eAeB}
		e^{\mathcal{A}}e^{\mathcal{B}}\phi=(e^{\mathcal{A}+\mathcal{B}}\phi) \,e^{\frac{1}{2}[\mathcal{A},\mathcal{B}]}=((e^{\mathcal{B}+\mathcal{A}}\phi) \,e^{\frac{1}{2}[\mathcal{B},\mathcal{A}]})\,e^{[\mathcal{A},\mathcal{B}]}
=(e^{\mathcal{B}}e^{\mathcal{A}}\phi)\,e^{[\mathcal{A},\mathcal{B}]}.
		\end{equation}
		Using this, we can obtain that
		$$\D(\mathfrak{q})\D(\mathfrak{p})
=e^{\sum_{\tau=i,j,k}2\tau(\mathfrak{q}_{\tau}\wedge\mathfrak{p}_{\tau})}\,\D(\mathfrak{p})\D(\mathfrak{q}).$$
In this regard, we can say that the map $\mathfrak{q}\longmapsto\D(\mathfrak{q})$ is a projective representation of the additive Abelian group $\quat$, since the composition of operators $\D(\mathfrak{q}),\,\D(\mathfrak{p})$ produce a phase factor, $e^{-\sum_{\tau=i,j,k}\tau(\mathfrak{q}_{\tau}\wedge\mathfrak{p}_{\tau})}$.\\

In order to establish the square integrability and irreducibility of $\D(\qu)$ we shall require some preliminaries. These preliminary definitions hold true for complex and quaternionic Hilbert spaces. The citations are given for the complex theory, however we shall present it in the language of this manuscript.
\begin{definition}\label{cyclic} (Page 144-146, \cite{Ba})
	A  representation $\D$ of $\HI$ is said to be cyclic if there exists $\phi\in\HI$ (called cyclic vector of $\D$) if
	$$\overline{\text{span}}\{\D(\qu)\phi | \qu\in \quat\}=\HI.$$
\end{definition}
\begin{lemma} (Lemma 12.1.3 in \cite{An}) Let $(\D, \HI, \hat{\mathfrak{A}}) $ be a representation. Suppose that $\phi\in\HI$ is a cyclic square integrable vector of the unitary representation $(\D, \HI)$. Then $\D$ is a square integrable representation.
\end{lemma}
\begin{definition}\label{SQ} (Page 204 \cite{Ali}) A representation $\D(\qu)$ is said to be square integrable if there exist a vector $\phi$ such that
	$$\int_{\quat} |\D(\qu)\phi\rangle\langle \D(\qu)\phi| d\varsigma=\mathbb{I}_{\mathfrak{H}}.$$
\end{definition}
\begin{proposition} (Page 146, \cite{Ba})
	A unitary representation $\D$ of $\hat{\mathfrak{A}}$ in $\HI$ is irreducible if and only if every non-zero vector $\phi\in\HI$ is cyclic for $\D$.
\end{proposition}
\begin{proposition} (Page 38, \cite{S})
	Every admissible vector is cyclic.
\end{proposition}
\begin{definition}\cite{Ali}\label{admissible}
	A vector $\eta\in\mathfrak{H}$ is said to be admissible for $\D(\qu)$, if
	\begin{equation}\label{adm_vec}
	I(\eta)=\int_{\quat}\mid\langle \D(\qu)\eta\mid\eta\rangle\mid^{2} d\varsigma<\infty.
	\end{equation}
\end{definition}
Now let us turn our attention to our case.
\begin{lemma}\label{Lem-admissible}
	If $\eta\in\mathfrak{H}$ is a admissible vectors, then
	so also is $\eta_\pu=\D(\pu)\eta$, for all $\pu\in\quat$.
\end{lemma}
\begin{proof}
For any $\qu,\pu\in\quat$,
\begin{equation}\label{ncomD}
\D(\mathfrak{q})\D(\mathfrak{p})
=e^{\sum_{\tau=i,j,k}2\tau(\mathfrak{q}_{\tau}\wedge\mathfrak{p}_{\tau})}\,\D(\mathfrak{p})\D(\mathfrak{q}).
\end{equation}
Now, let $\qu,\pu\in\quat$, then
	\begin{eqnarray*}
		I(\eta_{\pu})
		&=&\int_{\quat}\mid\langle \D(\qu)\eta_\pu\mid\eta_{\pu}\rangle\mid^{2}d\varsigma(r,\theta,\phi,\psi)\\
		&=& \int_{\quat}\mid\langle \D(\qu)\D(\pu)\eta\mid \D(\pu)\eta\rangle\mid^{2}d\varsigma(r,\theta,\phi,\psi)\mbox{~~as~~}\eta_\pu=\D(\pu)\eta\\
		&=& \int_{\quat}\mid\langle (\D(\pu)\D(\qu)\eta)\,e^{\sum_{\tau=i,j,k}2\tau(\mathfrak{q}_{\tau}\wedge\mathfrak{p}_{\tau})}\mid \D(\pu)\eta\rangle\mid^{2}d\varsigma(r,\theta,\phi,\psi)\mbox{~~by \ref{ncomD}}\\
		&=& \int_{\quat}\mid \,e^{-\sum_{\tau=i,j,k}2\tau(\mathfrak{q}_{\tau}\wedge\mathfrak{p}_{\tau})}\mid\,\mid\langle (\D(\pu)\D(\qu)\eta)\mid \D(\pu)\eta\rangle\mid^{2}d\varsigma(r,\theta,\phi,\psi)\\
		&=& \int_{\quat}\mid\langle \D(\qu)\eta\mid \eta\rangle\mid^{2}d\varsigma(r,\theta,\phi,\psi)\mbox{~~as~~}\D(\pu)\eta\mbox{~~is unitary and~~}\mid e^{-\sum_{\tau=i,j,k}2\tau(\mathfrak{q}_{\tau}\wedge\mathfrak{p}_{\tau})}\mid=1\\
		&=&I(\eta)\leq\infty.
	\end{eqnarray*}
	This concludes the proof.
\end{proof}

We have $\D(\qu)|e_0\rangle=|\qu\rangle$ and the resolution of the identity
$$\int_{\quat} |\qu\rangle\langle\qu|d\mu=\mathbb{I}_{\mathfrak{H}}.$$
Therefore
\begin{equation}\label{x1}
\int_{\quat}|\D(\qu)e_0\rangle\langle \D(\qu)e_0| d\mu=\mathbb{I}_{\mathfrak{H}}
\end{equation}
Therefore by Definition \ref{SQ} the representation $\D(\qu)$ is square integrable. From (\ref{x1}) we also have
$$\langle e_0| e_0\rangle=\int_{\quat} \langle e_0|\D(\qu)e_0\rangle\langle \D(\qu)e_0|e_0\rangle d\mu.$$
Since $|e_0\rangle\not=0,$ this gives us
$$0<\int_{\quat} |\langle e_0|\D(\qu)e_0\rangle|^2 d\mu=|~|e_0\rangle|^2<\infty,$$
therefore, $|e_0\rangle$ is an admissible vector (see definition (\ref{admissible}). Therefore, by Lemma \ref{Lem-admissible}, the set
$$\Lambda=\{\D(\qu)|e_0\rangle~|~\qu\in\quat\}$$
is a set of admissible vectors. On the other hand, the set $\Lambda$ is the set of all coherent states, so the set of vectors $\Lambda$ is not only total but is an overcomplete family in $\HI$ \cite{Bar}.

\begin{proposition} The representation $\D(\qu)$ is irreducible.
\end{proposition}
\begin{proof}
	Let $\phi\in\HI$ be any vector such that
	$$\langle \D(\qu)e_0|\phi\rangle=0, \quad \forall \qu\in\quat.$$
	Then from (\ref{x1}) we have
	$$\int_{\quat} \langle \phi|\D(\qu)e_0\rangle\langle \D(\qu)e_0|\phi\rangle d\mu=\|\phi\|^2=0$$
	which implies $\phi=0$. Hence $\Lambda^{\perp}=\{0\}$. That is, for each $\qu\in\quat$,  $\D(\qu)|e_0\rangle$ is cyclic and $\Lambda$ is dense in $\HI$. Therefore $\D(\qu)$ is irreducible.
\end{proof}
Therefore, in complete analogy with the complex displacement operator, we have established a square integrable, unitary and irreducible displacement operator in the quaternionic case.
\section{Symmetry}
In quantum mechanics (complex or quaternion) Hamiltonians are expected to be invariant under a unitary transformation. In this section, let us show this invariance for the quaternionic Hamiltonian $\hat{H}_I$ with the parity operator. In this section, we shall also discuss some properties of the displacement operator.
\subsection{Symmetry properties}
Consider the parity operator $\Pi=e^{\tau \pi a^{\dagger}a}=e^{\tau \pi N}$, $\tau\in\{i,j,k\}$ acting on $\HI$ as a linear operator. Note that
$$e^{\tau\pi N}|e_n\rangle=e^{\tau\pi n}|e_n\rangle=(\cos(\pi n)+\tau\sin(n\pi))|e_n\rangle=(-1)^n|e_n\rangle.$$
That is,
$$\Pi|e_n\rangle=(-1)^n|e_n\rangle$$
or equivalently
$$\Pi=\sum_{n=0}^{\infty}(-1)^n |e_n\rangle\langle e_n|.$$
Let us verify some properties of $\Pi$ and symmetry of operators.
\begin{enumerate}
	\item $\Pi^2|e_n\rangle=(-1)^n\Pi|e_n\rangle=(-1)^{2n}|e_n\rangle=\mathbb{I}_{\mathfrak{H}}|e_n\rangle$. Therefore
	$$\Pi^2=\mathbb{I}_{\mathfrak{H}}.$$
	That is, $\Pi$ is a unitary operator.
	\item
	\begin{eqnarray*}
		\Pi \as\Pi|e_n\rangle&=&(-1)^n\Pi \as|e_n\rangle=\sqrt{n}(-1)^n\Pi|e_{n-1}\rangle=(-1)^{2n-1}\sqrt{n}|e_{n-1}\rangle\\
		&=&-\sqrt{n}|e_{n-1}\rangle=-\as|e_n\rangle.
	\end{eqnarray*}
	Therefore
	$$\Pi \as\Pi=-\as.$$
	Similarly we have
	$$\Pi \asd\Pi=-\asd,\quad \Pi P_{\tau}\Pi=-P_{\tau},\quad \Pi Q\Pi=-Q.$$
	In other words, in terms of Poisson brackets,
	$$\{\Pi, \as\}=0,\quad\{\Pi, \asd\}=0,\quad\{\Pi, P_{\tau}\}=0,\quad\{\Pi, Q\}=0.$$
	\item[(3)] A straight forward calculation shows that
	$\Pi Q^2\Pi|e_n\rangle=Q^2|e_n\rangle$ and $\Pi P_I^2\Pi|e_n\rangle=P_I^2|e_n\rangle$. That is,
	$$\Pi Q^2\Pi=Q^2\quad\text{and}\quad \Pi P_I^2\Pi=P_I^2,$$
	and therefore $$\Pi \hat{H}_I\Pi=\hat{H}_I, $$ where $\hat{H}_I$ is the Hamiltonian as in equation (\ref{xxx}). That is, nondegenerate eigenkets of the Hamiltonians are also parity eigenkets.
	\end{enumerate}
Next result states how the parity operator acts on the displacement operator:
\begin{lemma} We have $\Pi \D(\mx)\Pi=\D(-\mx)$.
\end{lemma}
\begin{proof}
	We have $\D(\mx)=e^{-|\mx|^2/2}e^{\mx\cdot  \asd}e^{-\overline{\mx}\cdot \as}$.
	Since $$(\mx\cdot \asd)^{\dagger}={\asd}^{\dagger}\cdot\overline{\mx}=\as \cdot\overline{\mx}=\overline{\mx}\cdot\as$$
	and similarly $(\overline{\mx}\cdot \as)^{\dagger}=\mx\cdot \asd$ we have
	$$\D(\mx)^{\dagger}=\D(-\mx)=e^{-|\mx|/2}e^{-\mx\cdot  \asd}e^{\overline{\mx}\cdot \as}.$$
	Also note that $(\mx\cdot \as)^n=\mx^n\cdot (\as)^n.$ Since $(\overline{\mx}\cdot \as)^{n+1}|e_n\rangle=0$, we have
	$$e^{\overline{\mx}\cdot \as}|e_n\rangle=\sum_{j=0}^n\frac{\overline{\mx}^j}{j!}\sqrt{\frac{n!}{(n-j)!}}\cdot |e_{n-j}\rangle$$	
	and therefore
	$$e^{-\mx\cdot \asd}e^{\overline{\mx}\cdot \as}|e_n\rangle
	=\sum_{j=0}^n\sum_{m=0}^{\infty}\frac{(-1)^m\overline{\mx}^j\mx^m}{m!j!}\sqrt{\frac{n!(n+m-j)!}{[(n-j)!]^2}}
	\cdot |e_{n+m-j}\rangle.$$
	Therefore
	$$\D(-\mx)|e_n\rangle=e^{-|\mx|^2/2}	 \sum_{j=0}^n\sum_{m=0}^{\infty}\frac{(-1)^m~\overline{\mx}^j\mx^m}{m!j!}\sqrt{\frac{n!(n+m-j)!}{[(n-j)!]^2}}
	\cdot |e_{n+m-j}\rangle.$$
	Reasoning in a similar way we have
	\begin{eqnarray*}
		\Pi \D(\mx)\Pi |e_n\rangle&=&(-1)^n\Pi \D(\mx)|e_n\rangle\\
		&=&(-1)^n\Pi e^{-|\mx|^2/2}	 \sum_{j=0}^n\sum_{m=0}^{\infty}\frac{(-1)^j~\overline{\mx}^j\mx^m}{m!j!}\sqrt{\frac{n!(n+m-j)!}{[(n-j)!]^2}}
		\cdot |e_{n+m-j}\rangle.\\
		&=&(-1)^ne^{-|\mx|^2/2}	 \sum_{j=0}^n\sum_{m=0}^{\infty}\frac{(-1)^j(-1)^{n+m-j}~\overline{\mx}^j\mx^m}{m!j!}\sqrt{\frac{n!(n+m-j)!}{[(n-j)!]^2}}
		\cdot |e_{n+m-j}\rangle.\\
		&=&e^{-|\mx|^2/2}	 \sum_{j=0}^n\sum_{m=0}^{\infty}\frac{(-1)^m~\overline{\mx}^j\mx^m}{m!j!}\sqrt{\frac{n!(n+m-j)!}{[(n-j)!]^2}}
		\cdot |e_{n+m-j}\rangle.
	\end{eqnarray*}
	Therefore $\Pi \D(\mx)\Pi |e_n\rangle=\D(-\mx)|e_n\rangle$ for all $n$. Hence  $\Pi \D(\mx)\Pi=\D(-\mx)$.
\end{proof}
\subsection{Some properties of the displacement operator}
The following proposition discusses two versions for the displacement operator. We note that in the complex case the normal and the anti-normal ordering are both used without making any distinction between them.
\begin{proposition}
	The displacement operator $\D(\qu)$ satisfies
	\begin{itemize}
		\item [(i)] the normal ordering property:  $\D(\qu)=e^{-\frac{|\qu|^2}{2}}e^{\qu\cdot\mathsf{a}^\dagger}e^{-\oqu\cdot\mathsf{a}}$,
		\item [(ii)] the anti-normal ordering property:  $\D(\qu)=e^{\frac{|\qu|^2}{2}}e^{-\oqu\cdot\mathsf{a}}e^{\qu\cdot\mathsf{a}^\dagger}$.
	\end{itemize}
	Furthermore, the coherent state $\mid\gamma_{\qu}\rangle$ is generated from the ground state $\mid e_0\rangle$ by the displacement operator $\D(\qu)$,
	\begin{equation}\label{coh_dis}
	\mid\gamma_{\qu}\rangle=\D(\qu)\mid e_0\rangle.
	\end{equation}
\end{proposition}
\begin{proof}
	Now, using Baker-Campbell-Hausdorff identity, we have $$e^{\qu\cdot\mathsf{a}^\dagger}e^{-\oqu\cdot\mathsf{a}}=e^{\frac{1}{2}[\qu\cdot\mathsf{a}^\dagger,-\oqu\cdot\mathsf{a}]}e^{\qu\cdot\mathsf{a}^\dagger-\oqu\cdot\mathsf{a}}~\mbox{~and~}~e^{\qu\cdot\mathsf{a}^\dagger}e^{-\oqu\cdot\mathsf{a}}=e^{\frac{1}{2}[-\oqu\cdot\mathsf{a},\qu\cdot\mathsf{a}^\dagger]}e^{\qu\cdot\mathsf{a}^\dagger-\oqu\cdot\mathsf{a}}.$$
	This implies,
	$$\D(\qu)=e^{-\frac{1}{2}[\qu\cdot\mathsf{a}^\dagger,-\oqu\cdot\mathsf{a}]}e^{\qu\cdot\mathsf{a}^\dagger}e^{-\oqu\cdot\mathsf{a}}=e^{-\frac{1}{2}[-\oqu\cdot\mathsf{a},\qu\cdot\mathsf{a}^\dagger]}e^{-\oqu\cdot\mathsf{a}}e^{\qu\cdot\mathsf{a}^\dagger}.$$
	But, if $\qu=q_0+iq_i+jq_j+kq_k$, then
	\begin{eqnarray*}
		[\qu\cdot\mathsf{a}^\dagger,-\oqu\cdot\mathsf{a}]
		&=&-\left( \sum_{\tau=i,j,k}\left[(\dfrac{q_0}{\sqrt{3}}+q_\tau\,\tau)(\dfrac{q_0}{\sqrt{3}}-q_\tau\,\tau) \right]\right)\cdot\mathbb{I}_{\mathfrak{H}}\\
		&=&-\left( \sum_{\tau=i,j,k}\left[\dfrac{{q_0}^2}{3}-q_\tau^2\,\tau^2 \right]\cdot\mathbb{I}_{\mathfrak{H}}\right) =-(q_0^2+q_i^2+q_j^2+q_k^2)\mathbb{I}_{\mathfrak{H}}=-|\qu|^2\mathbb{I}_{\mathfrak{H}}.
	\end{eqnarray*}
	Thus $\D(\qu)=e^{-\frac{|\qu|^2}{2}}e^{\qu\cdot\mathsf{a}^\dagger}e^{-\oqu\cdot\mathsf{a}}=e^{\frac{|\qu|^2}{2}}e^{-\oqu\cdot\mathsf{a}}e^{\qu\cdot\mathsf{a}^\dagger}$, that is, (i) and (ii) follow. From the fact that $\mathsf{a}\mid e_0\rangle=0$ and the normal ordering property: $\D(\qu)=e^{-\frac{|\qu|^2}{2}}e^{\qu\cdot\mathsf{a}^\dagger}e^{-\oqu\cdot\mathsf{a}}$, we get
	$$\D(\qu)\mid e_0\rangle=e^{-\frac{|\qu|^2}{2}}e^{\qu\cdot\mathsf{a}^\dagger}e^{-\oqu\cdot\mathsf{a}}\mid e_0\rangle=e^{-\frac{|\qu|^2}{2}}e^{\qu\cdot\mathsf{a}^\dagger}\mid e_0\rangle=\mid\gamma_{\qu}\rangle.$$  The claims get proved.
\end{proof}
In the following proposition we show a peculiarity of the displacement operator, namely the analog of the covariance property in the quaternionic setting:
\begin{proposition}
	For any $\qu,\pu\in\quat$, $\D(\qu)\D(\pu)\D(\qu)^{\dagger}=e^{-\sum_{\tau=i,j,k}2\tau(\mathfrak{q}_{\tau}\wedge\mathfrak{p}_{\tau})}\,\D(\pu)$.
\end{proposition}
\begin{proof}
	Let $\qu,\pu\in\quat$ and $\mathcal{A}=(\mathfrak{q}\cdot \mathsf{a}^\dagger-\oqu\cdot \mathsf{a})$ and $\mathcal{B}=(\mathfrak{p}\cdot \mathsf{a}^\dagger-\bar{\mathfrak{p}}\cdot \mathsf{a})$. Firstly, it should be noticed that, from Baker-Campbell-Hausdorff identity, we get
	\begin{eqnarray*}
		\D(\qu)[\D(\pu)\D(\qu)^{\dagger}]\phi
		&=&e^{A}[e^{B}e^{-A}]\phi=e^{A}[e^{\frac{1}{2}[B,-A]}e^{B-A}]\phi\\
		&=&e^{A}[e^{B-A}\phi\, e^{\frac{1}{2}[B,-A]}]=e^{B}(\phi\, e^{\frac{1}{2}[B,-A]})e^{\frac{1}{2}[A,B-A]}\\
		&=&e^{B}(\phi\, e^{[A,B]})=[e^{-\sum_{\tau=i,j,k}2\tau(\mathfrak{q}_{\tau}\wedge\mathfrak{p}_{\tau})}\,\D(\pu)]\phi.
	\end{eqnarray*}
Hence,
 $$\D(\qu)\D(\pu)\D(\qu)^{\dagger}=e^{-\sum_{\tau=i,j,k}2\tau(\mathfrak{q}_{\tau}\wedge\mathfrak{p}_{\tau})}\,\D(\pu),$$ the desired result.
\end{proof}
In the complex case, the displacement operator for one mode in quantum optics is the shift operator $\D(z)=e^{z\asd-\oz\as}$, where $z$ is the amount of displacement in the optical phase space. The following proposition validate this claim in the quaternion case.
\begin{proposition} The displacement operator $\D(\mx)$ satisfies the following properties.
	$$(i)~~\D(\mx)^{\dagger}\as \D(\mx)=\as+\mx.\quad\quad
	(ii)~~\D(\mx)^{\dagger}\asd \D(\mx)=\asd+\overline{\mx}.$$
\end{proposition}
\begin{proof}
	$(i)$~Since $$[\as, (\asd)^2] |e_n\rangle=2\sqrt{n+1} |e_{n+1}\rangle=2\asd |e_n\rangle\quad\text{and}\quad
	[\as^2, \asd] |e_n\rangle=2\sqrt{n} |e_{n-1}\rangle=2\as |e_n\rangle,$$ we have
	$$[\as, (\asd)^2]=2\asd\quad\text{and}\quad [\as^2, \asd]=2\as.$$
	Similarly, when $n$ is a positive integer, by induction on $n$ we get
	$$[\as, (\asd)^n]=n (\asd)^{n-1} \quad\text{and}\quad [\as^n, \asd]=n \as^{n-1}.$$
	Consider
	\begin{eqnarray*}
		[\as, e^{\mx\cdot\asd}]&=&\left[\as, \sum_{n=0}^{\infty}\frac{(\mx\cdot\asd)^n}{n!}\right]\\
		&=&\sum_{n=0}^{\infty}\frac{\mx^n}{n!}\cdot [\as, (\asd)^n]\\
		&=&\sum_{n=1}^{\infty}\frac{\mx^n}{n!}\cdot n (\asd)^{n-1}\quad\text{as}\quad [\as, \mathbb{I}_\HI]=0\\
		&=&\mx\sum_{n=1}^{\infty}\frac{\mx^{n-1}}{(n-1)!}\cdot (\asd)^{n-1}\\
		&=&\mx\sum_{n=0}^{\infty}\frac{\mx^n}{n!}\cdot(\asd)^n=\mx e^{\mx\cdot\asd}.	
	\end{eqnarray*}
	That is,
	$$\as e^{\mx\cdot\asd}-e^{\mx\cdot\asd}\as=\mx e^{\mx\cdot\asd}.$$
	Hence multiplying by $e^{-\mx\cdot\asd}$ we get
	$$e^{-\mx\cdot\asd}\as e^{\mx\cdot\asd}-\as=\mx$$
	or $e^{-\mx\cdot\asd}\as e^{\mx\cdot\asd}=\as+\mx$ which implies
	$$e^{\overline{\mx}\cdot\as}e^{-\mx\cdot\asd}\as  e^{\mx\cdot\asd}e^{-{\overline{\mx}}\cdot\as}=\as+\mx.$$
	Since $e^{-|\mx|^2/2}$ is a real number the last expression implies
	$$e^{|\mx|^2/2}e^{\overline{\mx}\cdot\as}e^{-\mx\cdot\asd}\as e^{-|\mx|^2/2} e^{\mx\cdot\asd}e^{-{\overline{\mx}}\cdot\as}=\as+\mx$$
	or equivalently
	$$\D(\mx)^{\dagger}\as \D(\mx)=\as+\mx,$$
	where we have used the normal and anti-normal orderings of $\D(\mx)$, that is, we have used
	$$\D(\mx)=e^{\mx\cdot\asd-\overline{\mx}\cdot\as}=e^{-|\mx|^2/2}e^{\mx\cdot\asd}e^{-\overline{\mx}\cdot\as}=e^{-\overline{\mx}\cdot\as}e^{\mx\cdot\asd}e^{|\mx|^2/2}.$$
	The identity $(ii)$ is simply the adjoint of $(i)$.
\end{proof}
By Proposition \ref{Pr1}, any non-real quaternion $\mx$ can be written in the form $\mx=x+\tau y$ where $x=q_0$, $y=|q_1i+q_2j+q_3k|$ and $\tau=q_1i+q_2j+q_3k/|q_1i+q_2j+q_3k|\in\mathbb S$. By $\D_{\tau}(\qu)$ we denote the displacement operator restricted to the quaternion slice $\mathbb C_{\tau}$. We have the following proposition.
\begin{proposition}\label{Pro:4.16} The displacement operator satisfies the following properties.
	$$(i)~~\frac 12\left(\frac{\partial}{\partial x}-\tau \frac{\partial}{\partial y}\right)\D_{\tau}(\qu)=(\as^{\dagger}-\frac 12 \bar{\qu} ) \D_{\tau}(\qu).\quad\quad
	(ii)~~\frac 12\left(\frac{\partial}{\partial x}+\tau \frac{\partial}{\partial y}\right)\D_{\tau}(\qu)=-(\as-\frac 12 \qu) \D_{\tau}(\qu)$$
\end{proposition}
\begin{proof}
To show that $(i)$ holds we use Proposition \ref{xAq}, namely the fact that for any quaternion $\qu$ we have $\qu\cdot \as=\as \cdot \qu$ and $\qu\cdot \as^\dagger=\as^\dagger \cdot \qu$.
 \[
 \begin{split}
 &\frac 12\left(\frac{\partial}{\partial x}-\tau \frac{\partial}{\partial y}\right)(e^{-\frac{|\qu|^2}{2}}e^{\qu\cdot\mathsf{a}^\dagger}e^{-\bar \qu\cdot\mathsf{a}})\\
 &=\frac 12 \left(e^{-\frac{|\qu|^2}{2}}(-x)e^{\qu\cdot\mathsf{a}^\dagger}e^{-\qu\cdot\mathsf{a}}+e^{-\frac{|\qu|^2}{2}}
 e^{\qu\cdot\mathsf{a}^\dagger}\as^\dagger
 e^{-\qu\cdot\mathsf{a}}-e^{-\frac{|\qu|^2}{2}}e^{\qu\cdot\mathsf{a}^\dagger}e^{-\qu\cdot\mathsf{a}}\as\right.\\
 &\left. -\tau e^{-\frac{|\qu|^2}{2}}(-y)e^{\qu\cdot\mathsf{a}^\dagger}e^{-\qu\cdot\mathsf{a}}-\tau e^{-\frac{|z|^2}{2}}
 e^{\qu\cdot\mathsf{a}^\dagger}(\tau \as^\dagger)
 e^{-\qu\cdot\mathsf{a}}-\tau e^{-\frac{|\qu|^2}{2}}e^{\qu\cdot\mathsf{a}^\dagger}e^{-\qu\cdot\mathsf{a}}(\tau \as)\right)\\
 &= \D_{\tau}(\qu) (\as^\dagger -\frac 12 \bar \qu),
 \end{split}
 \]
 where, to show the equality, we used the obvious fact that $\tau$ commutes with $\qu\in\mathbb C_{\tau}$ and thus it can be moved to the left side. The second relation can be proved in a similar way.
\end{proof}
\begin{remark} We point out that in Proposition \ref{Pro:4.16} the crucial facts which make the proof work are the validity of Proposition \ref{xAq} and the fact that the imaginary unit $\tau$ commutes with the variable $\qu\in\mathbb C_{\tau}$. One may ask if the properties in Proposition \ref{Pro:4.16} hold by considering the operators
$\frac{\partial}{\partial \qu}$ and $\frac{\partial}{\partial \oqu}$ and the displacement operator $\D(\qu)$. To show that the answer is negative, let us consider
$$
\frac{\partial}{\partial \bar{\qu}}=\frac{\partial}{\partial x_0} +i\frac{\partial}{\partial x_1}+j\frac{\partial}{\partial x_2}+k\frac{\partial}{\partial x_3}
$$
which is the so-called Cauchy-Fueter operator. We observe that
$$
\frac{\partial}{\partial x_1}\qu^n=i\qu^{n-1}+\qu i\qu^{n-2}+\qu^2i\qu^{n-3}+\cdots+\qu^{n-1}i
$$
and similarly for the other partial derivatives. It is obvious that the three imaginary units do not commute with the variable $\qu$ (unlike what happens above for the imaginary unit $\tau$ and the variable on the complex plane $\mathbb C_{\tau}$). The same problem arises when applying the partial derivatives to $(\qu\cdot \as^{\dagger})^n$ and, more in
general, to $\sum_{n=0}^\infty \frac{1}{n!}(\qu\cdot \as^\dagger)^n$. Thus one cannot obtain an expression for $\frac{\partial}{\partial \bar{\qu}} \D(\qu)$. This fact shows that in the slice-wise approach we can prove the validity of additional properties of the displacement operator while, in the global approach, the issues related to the non-commutative setting still appear.
\end{remark}
It is worthwhile to add a comment on the notion on analyticity underlying this new approach to quaternionic quantum mechanics, even though we do not explicitly make use of this notion. However, the issue of what is the appropriate notion for analyticity has been pointed out at the end of the introduction of Adler's book \cite{Ad}. The theory of functions in the kernel of the Cauchy-Fueter operator $\frac{\partial}{\partial \oqu}$, though very important and very well developed, is not appropriate for several reasons. For example, the identity functions, functions of the form $f(\qu)=\qu^n$, the exponential $e^{\qu}$ are not in the kernel of the Cauchy-Fueter operator. Thus, in the corresponding functional calculus, we cannot define $e^T$ where $T$ is a  linear operator. Moreover, the functional calculus does not allow a suitable definition of spectrum and, consequently, a spectral theorem. We refer the interested reader to \cite{NFC} for more information. For these reasons, it seems better to use a different function theory and differential operators different from  $\frac{\partial}{\partial \oqu}$ or $\frac{\partial}{\partial {\qu}}$. The notion of analyticity used in \cite{NFC} and several other subsequent works was the notion of slice hyperholomorphy, see \cite{Albook, NFC, Gen1}. As we already explained, any non-real quaternion $\qu$ can be written uniquely in the form $\qu=x+\tau y$; if $\qu$ is real then it is trivial that $\qu=x+\tau 0$ for any $\tau$. A real differentiable function $f(\qu)$ is slice hyperholomorphic if $f(\qu)=\alpha(x,y)+\tau \beta(x,y)$ where $\alpha$ and $\beta$ satisfy the Cauchy-Riemann system $\partial_x\alpha -\partial_y\beta=0$, $\partial_y\alpha +\partial_x\beta=0$ and, in order to have a function well-posed, $\alpha(x,-y)=\alpha(x,y)$, $\beta(x,-y)=-\beta(x,y)$. This class of functions includes the powers $f(\qu)=\qu^n$ and the exponential $e^{\qu}$. The functional calculus associated with it is based on the $S$-spectrum. One should also note that any normal operator $T$ can be written in the form $T=A+JB$ where $A$ is self-adjoint, $B$ is positive and $J$ is anti self-adjoint and unitary. Thus the class of functions allows to define $f(T)=\alpha(A,B)+J \beta(A,B)$.
Via the $S$-spectrum, which naturally arises from the construction of the functional calculus, one also can prove the spectral theorem for normal operators. Thus, also from the point of view of the underlying function theory, we believe that quaternionic quantum mechanics has now solid mathematical grounds.
\section{Conclusion}
 In the quaternionic quantum mechanics literature there were several attempts to define a universal linear and self-adjoint momentum operator in quaternionic Hilbert spaces resembling the complex momentum operator. It became obvious that, due to the non-commutativity of quaternions, such an attempt cannot flourish in a left quaternionic Hilbert space with a left multiplication only or in a right quaternionic Hilbert space with a right multiplication only (see \cite{Ad} for a complete review on this point). However, in this manuscript, we have demonstrated that if we consider a right quaternionic Hilbert space with a left multiplication on it, such an operator can be defined. Using the so defined operators we have studied the Heisenberg uncertainty principle. A parallel picture of the complex harmonic oscillator was not achieved on the whole space of quaternions. This difficulty is not due to the defined operators but to the fact that it was not possible to compute the sum of a quaternionic series in a closed form. However, a satisfactory description has been obtained on a quaternionic slice. \\

A quaternionic displacement operator comparable to that one of the complex harmonic oscillator has been attempted a few times. While studying Perelomov type coherent states in the quaternionic setting, in \cite{Ad2} the authors proved that a unitary and irreducible displacement operator cannot be obtained. In \cite{Thi2}, while studying quaternionic canonical coherent states, it is shown that a displacement operator like the complex harmonic oscillator can generate coherent states from the ground state, however, it cannot be associated to any algebra as a representation. In fact, it was proved that there is no such an algebra which is closed (over the field in which it is spanned). However, in this manuscript, we have proved that, a displacement operator resembling the complex harmonic oscillator can be obtained as a representation on a Weyl-Heisenberg type algebra with all the desired properties.\\

Thus, in our opinion, 80 years after the paper of Birkhoff and von Neumann, quaternionic quantum mechanics is now equipped with almost all the necessary tools to develop it.

%

\begin{thebibliography}{XXXX}
	
	\bibitem{Ad} Adler, S.L., {\em Quaternionic quantum mechanics and Quantum fields}, Oxford University Press, New York, 1995.
	
	\bibitem{Ad2} Adler, S.L., Millard, A.C.,{\em Coherent states in quaternionic quantum mechanics}, J. Math. Phys.,{\bf 38} (1997), 2117-26.
	
	\bibitem{Ali} Ali, S.T., Antoine, J-P., Gazeau, J-P., {\em Coherent States, Wavelets and Their Generalizations}, Second edition, Springer,  New York, 2014.
	
	\bibitem{AE} Ali, S.T., Englis, M., {\em Quantization method: a quide for Physicist and analysts}, Rev. Math. Phys., {\bf 17} (2005), 391-490.
	
\bibitem{ack}
{Alpay} D., {Colombo} F.,  Kimsey D. P.,
{\em The spectral theorem for  quaternionic unbounded normal operators based on the $S$-spectrum},
J. Math. Phys. {\bf 57} (2016), 023503.

\bibitem{Albook} Alpay, D., Colombo, F., Sabadini, I., {\em Slice Hyperholomorphic Schur Analysis}, Operator Theory: Advances and Applications 256, Birkh\"auser 2016.

\bibitem{acs} Alpay, D., Colombo, F., Sabadini, I., {\em On a class of quaternionic positive definite functions and their derivatives}, J. Math. Phys., {\bf 58} (2017), 033501.


	\bibitem{Al} Alpay, D., Colombo, F., Sabadini, I., Salomon, G., {\em The Fock space in the slice hyperholomorphic setting}, Hypercomplex Analysis: New perspective and applications, Trends in Mathematics, Birkh\"user, Basel (2014), 43-59.
	
	\bibitem{Bar} Bargmann, V., Butera, P., Girardello, L., Klauder, J.R., {\em On the completeness of coherent states}, Rep. Math. Phys. {\bf 2} , 221-228 (1971).
	
	\bibitem{Ba} Barut, A. O., Raczka, R, {\em Theory of group representations and applications}, Polish Scientific Publications, Warszawa, 1980.
	
	\bibitem{bvn}  Birkhoff G., von Neumann J., {\em The logic of quantum mechanics}, Ann. of Math.,
	37 (1936), 823--843.
	
	\bibitem{Brian}Brian C. Hall, {\em Quantum theory for Mathematicians}, Springer, New York (2013).
	
	
	\bibitem{NFC} Colombo, F., Sabadini, I., Struppa, D.C.,{\em Noncommutative Functional Calculus}, Birkh\"user Basel, Berlin, 2011.
	
	
	\bibitem{S} Harmonic and applied analysis: From groups to signals, Eds:Dahlke, S, De Mari, F., Grohs, P. Labate, D., Springer, Switzerland (2015).
	
	\bibitem{An} Di Anton Deitmar, Siegfried Echtehoff, {\em Principles of Harmonic Analysis}, Springer, LLC (2009).
	
	
	\bibitem{Eb} Ebbinghaus, H.D., Hermes, H., Hirzebruch, F., Koecher, M., Mainzer, K., Neukirch, J., Prestel,
	A., Remmert, R., {\em Numbers},(3rd Edition), Springer, New York (1995).
	
	\bibitem{fjss} Finkelstein, D.,  Jauch, J.M., Schiminovich, S. and Speiser, D.,
	{\em Foundations  of  quaternion quantum mechanics}, J. Math. Phys. 3
	(1962), 207--220.
	
	
	\bibitem{Gaz} Gazeau, J-P., {\em Coherent states in quantum physics}, Wiley-VCH,  Berlin (2009).
	
	
	\bibitem{Gen1} Gentili, G., Stoppato C., Struppa, D.C., {\em Regular Functions of a Quaternionic Variable}, Springer Monograph in Mathematics, Springer, 2013.
	
	\bibitem{ghimorper} Ghiloni, R., Moretti, W. and Perotti, A., {\em Continuous slice functional calculus in quaternionic Hilbert spaces\/,} Rev. Math. Phys. {\bf 25} (2013), 1350006.
	
	\bibitem{Gu} G\"urlebeck, K., Habetha, K., Spr\"ossig, W., {\em Holomorphic functions in the plane and n-dimensional spaces}, Birkh\"auser Verlag, Basel (2008).
	
	\bibitem{MuTh} Muraleetharan. B., Thirulogasanthar, K., {\em Coherent state quantization of quaternions}, J. Math. Phys., {\bf 56} (2015), 083510.
	
	\bibitem{NJ}  Nash C.G., Joshi, J., {\em Composite systems in quaternionic quantum mechanics}, J. Math. Phys., {\bf 28} (1987), 2883. 
	
	\bibitem{Thi1} Thirulogasanthar, K., Twareque Ali, S., {\em Regular subspaces of a quaternionic Hilbert
		space from quaternionic Hermite polynomials and associated coherent states}, J. Math. Phys., {\bf 54} (2013), 013506.
	
	\bibitem{Thi2} Thirulogasanthar, K., Honnouvo, G., Krzyzak, A., {\em Coherent states and Hermite polynomials on Quaternionic Hilbert spaces}, J. Phys.A: Math. Theor. {\bf 43} (2010), 385205.
	

\bibitem{Teich} Teichm\"uller, O., {\em Operatoren im Wachsschen Raum}, J. Reine Angew. Math. {\bf 174} (1936), 73--124.
	
	\bibitem{Vis} Viswanath, K., {\em Normal operators on quaternionic Hilbert spaces}, Trans. Am. Math. Soc. {\bf 162} (1971), 337--350.
	
	\bibitem{Za} Zhang, F., {\em Quaternions and Matrices of Quaternions}, Linear Algebra and its Applications, {\bf 251} (1997), 21-57.
	
\end{thebibliography}
\end{document}